\documentclass[letterpaper]{article} 
\usepackage{aaai2026}  
\usepackage{times}  
\usepackage{helvet}  
\usepackage{courier}  
\usepackage[hyphens]{url}  
\usepackage{graphicx} 
\usepackage{natbib}  
\usepackage{caption} 
\usepackage{algorithm}
\usepackage{algorithmic}
\usepackage{comment}
\usepackage{booktabs}
\usepackage{amsmath}
\usepackage{amssymb}
\usepackage{amsthm}
\newtheorem{definition}{Definition}
\newtheorem{proposition}{Proposition}
\newtheorem{theorem}{Theorem}
\newtheorem{lemma}{Lemma}
\usepackage{mathtools}

\usepackage{newfloat}
\usepackage{listings}
\DeclareCaptionStyle{ruled}{labelfont=normalfont,labelsep=colon,strut=off} 
\floatstyle{ruled}
\newfloat{listing}{tb}{lst}{}
\floatname{listing}{Listing}
\title{Revisiting the Canonicalization for Fast and Accurate
Crystal Tensor Property Prediction}

\author {
    Haowei Hua\textsuperscript{\rm 1},
    Jingwen Yang\textsuperscript{\rm 1},
    Wanyu Lin\textsuperscript{\rm 1}\thanks{Corresponding author.},
    Zhou Pan\textsuperscript{\rm 2}
}
\affiliations {
    \textsuperscript{\rm 1}The Hong Kong Polytechnic University\\
    \textsuperscript{\rm 2}Singapore Management University\\
    haowei.hua@connect.polyu.hk, jingwen.yang@connect.polyu.hk, 
    wan-yu.lin@polyu.edu.hk,
    panzhou@smu.edu.sg
}

\usepackage{bibentry}

\begin{document}

\maketitle
\begin{abstract}
Predicting the tensor properties of crystalline materials is a fundamental task in materials science. Unlike scalar property prediction, which requires invariance, tensor property prediction requires maintaining $O(3)$ group tensor equivariance. 
Achieving such equivariance typically demands specialized architectural designs, which substantially increase computational cost. Canonicalization, a classical technique for geometry, has recently been explored for efficient learning with symmetry.
In this work, we revisit the problem of crystal tensor property prediction through the lens of canonicalization. Specifically, we demonstrate how polar decomposition, a simple yet efficient algebraic method, can serve as a form of canonicalization and be leveraged to ensure equivariant tensor property prediction. Building upon this insight, we propose a general $O(3)$-equivariant framework for fast and accurate crystal tensor property prediction, referred to as {\em GoeCTP}. By utilizing canonicalization, {\em GoeCTP} achieves high efficiency without requiring the explicit incorporation of equivariance constraints into the network architecture.
Experimental results indicate that {\em GoeCTP} achieves the high prediction accuracy and runs up to 13$\times$ faster compared to existing state-of-the-art methods, underscoring its effectiveness and efficiency. 

\end{abstract}


\section{Introduction}

The tensor properties of crystalline materials can capture intricate material responses through high-order tensors,
with wide-ranging applications in fields such as physics, electronics, and engineering~\citep{yanspace}. These tensor properties span various orders, such as dielectric tensor with two orders, piezoelectric tensor with three orders, and elastic tensor with four orders.
Accurate prediction of these tensor properties is critical for novel materials discovery and design with targeted characteristics. 
Thus far, several works have been dedicated to crystal tensor property prediction. One prominent category involves {\em ab initio} physical simulation techniques, such as density functional theory (DFT) \citep{petousis2016benchmarking}. These classical simulation techniques can provide acceptable error margin for predicting various material properties. However, they necessitate extensive computational resources due to the complexity of handling large crystals with a vast number of atoms and electrons~\citep{yanspace}, hindering their applicability in practice. 

Alternatively, machine learning (ML) models have been proposed to facilitate the process of crystalline material property prediction. These methods typically leverage high-precision datasets deriving from {\em ab initio} simulations and utilize crystal graph construction techniques along with graph neural networks (GNNs) \citep{chen2019graph,louis2020graph,choudhary2021atomistic,xie2018crystal} or transformers~\citep{yancomplete,taniaicrystalformer,yan2022periodic,lee2024density,wang2024conformal,ito2025rethinking}.
Most existing methods are designed for scalar property prediction, focusing on achieving $SO(3)$ invariance of the crystal structures. However, these scalar-property methods could not be adapted to predict tensor properties of crystalline materials, which exhibits significantly higher complexity. 
This complexity arises from the fact that tensor properties describe how crystals respond to external physical fields, such as electric fields or mechanical stress
\citep{nye1985physical,resta1994macroscopic,yanspace}; their modeling necessitates preserving consistency with the crystal’s spatial orientation, exhibiting a unique tensor $O(3)$ equivariance \citep{yanspace}.

Therefore, a few recent studies attempt to ensure equivariance through specialized designs of network architectures~\citep{mao2024dielectric,lou2024discovery,heilman2024equivariant,wen2024equivariant,pakornchote2023straintensornet,yanspace, zhong2023general}. These methods generally integrate directional features and spherical harmonics to preserve equivariance.
However, achieving this often requires computationally expensive operations such as tensor products, which introduce substantial overhead, especially when dealing with high-order tensor data. {\em Therefore, providing fast and accurate predictions of tensor properties across various materials is challenging.} 

\begin{figure*}[t] 
  \centering   
  \includegraphics[width=\textwidth]{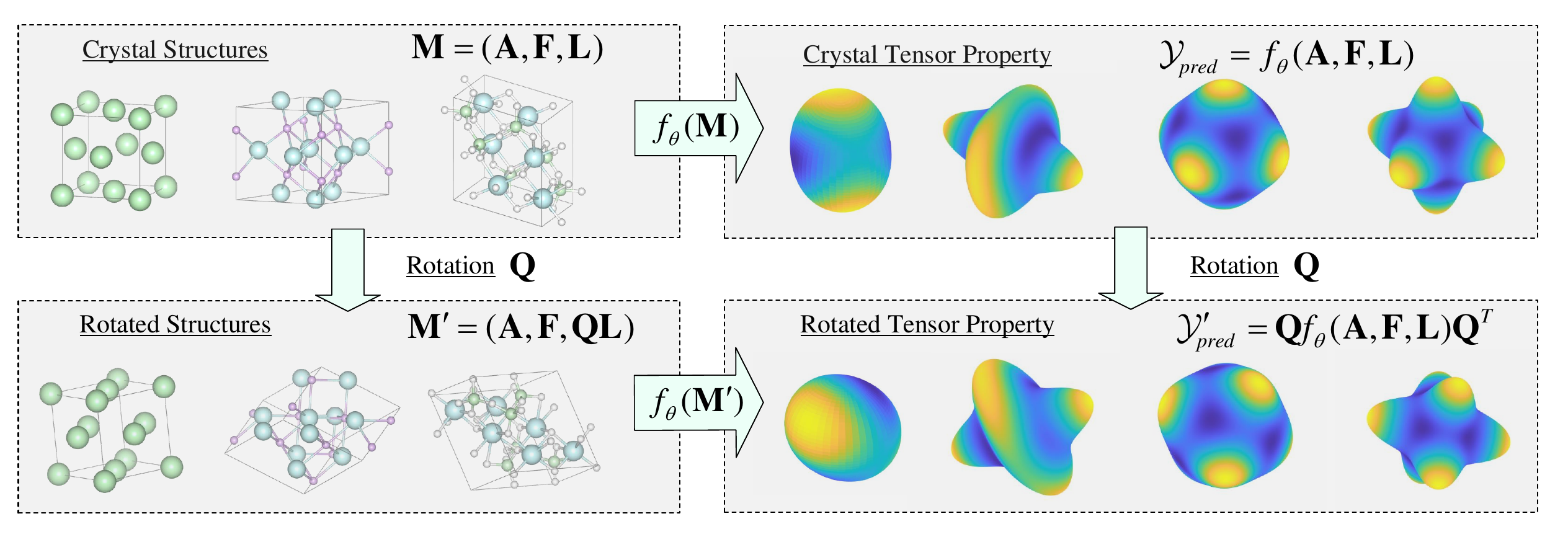}
  \caption{The illustration of $O(3)$-equivariance for crystal tensor prediction. The visualization of the crystal structures on the left was generated using VESTA \citep{momma2011vesta}, while the visualization of the crystal tensor properties on the right follows the method of \citet{yanspace} and VELAS \citep{ran2023velas}.
  } 
  \label{equvarience_tensor}
\end{figure*}

One approach that has show promising results in efficient learning with symmetry across various fields, including point clouds \citep{linequivariance}, 
n-body simulation \citep{kaba2023equivariance}, and antibodies generation~\citep{martinkus2023abdiffuser}, is the use of ``canonicalization". Specifically,  canonicalization maps a geometric data to an invariant representation \citep{ma2024canonicalization,dymequivariant}, referred to as the canonical form, and subsequently reconstructs an equivariant output from the canonical form without imposing any architectural constraints on the backbone network. However, to date, it has not been explored for crystal tensor property prediction.

In this work, we revisit the task of crystal tensor property prediction through the lens of canonicalization and introduce a novel canonicalization strategy tailored for this particular setting. Specifically, rather than embedding equivariance directly into model architecture, we propose a simple yet effective canonicalization module, termed R\&R. In particular, our canonicalization is instantiated based on polar decomposition, a continuous mapping technique~\citep{dymequivariant} that can provide enhanced robustness.
During prediction, the R\&R module applies rotations and reflections to transform the input crystal structure into its canonical form. The canonical form is then fed into a property prediction network to obtain the corresponding canonical tensor representation. Simultaneously, the orthogonal matrix derived from the R\&R module is utilized to recover the equivariant output via the tensor transformation rule, enabling equivariant tensor property prediction without incurring additional computational overhead.
We conducted experiments on dielectric, piezoelectric, and elastic tensor datasets, respectively, showcasing that the proposed method can achieve $O(3)$-equivariant tensor prediction while maintaining high efficiency. Compared to the previous state-of-the-art work~\citet{yanspace}, the {\em GoeCTP} method achieves high prediction accuracy and runs by up to 13$\times$ faster.

\begin{figure*}[t] 
  \centering   
  \includegraphics[width=\textwidth]{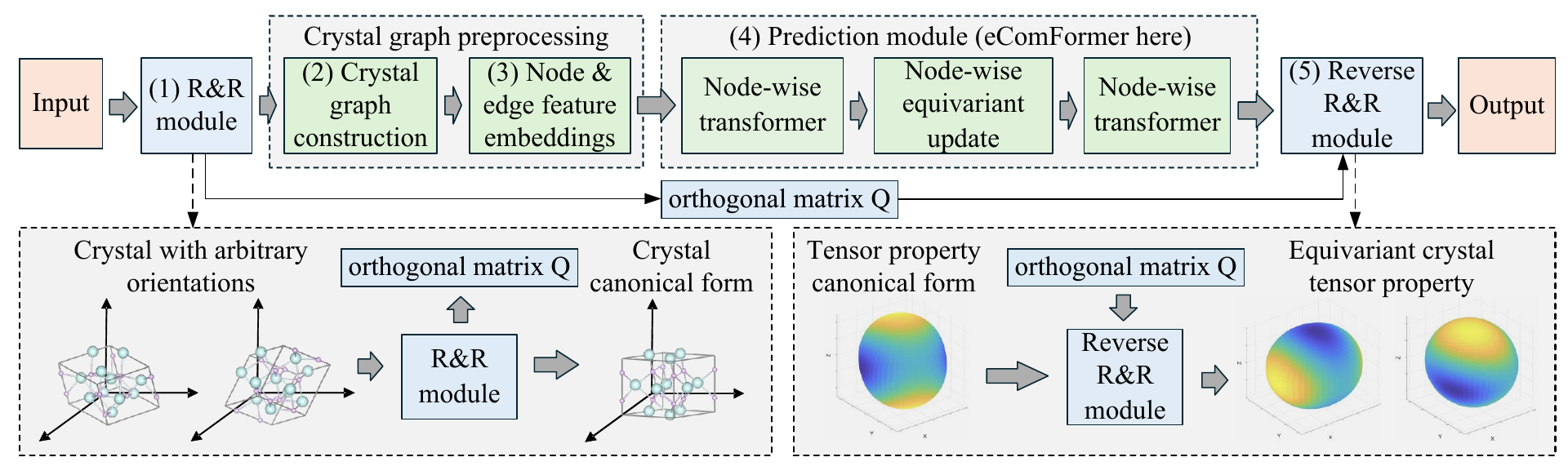}
  \caption{The Illustration of {\em GoeCTP}. To begin with, \textbf{(1)} the R\&R module rotates and reflects the input crystal structure, which may have an arbitrary orientation, to the 
canonical form of this crystal. Next, \textbf{(2)} Crystal Graph Construction module organizes the adjusted input into a crystal graph, followed by \textbf{(3)} the Node \& Edge Feature Embedding module, which encodes the features of the crystal graph. Subsequently, \textbf{(4)} the Prediction module leverages these embedded features to predict the 
canonical form of tensor properties corresponding to the 
canonical form of input crystal. Finally, \textbf{(5)} the Reverse R\&R module applies the orthogonal matrix $\mathbf{Q}$, obtained from the R\&R module, to ensure the equivariance of the output tensor properties.} 
  \label{fig:overview}
\end{figure*}

\section{Background}
\label{sec: pre}
\subsection{Preliminaries}
Crystalline materials consist of a periodic arrangement of atoms in 3D space, characterized by a repeating unit known as the unit cell.
A complete crystal can thus be described by the atomic types, atomic coordinates, and lattice parameters of a single unit cell~\citep{yan2022periodic,wang2024crystalline,jiaospace,huangcode,hua2025local}.
Mathematically, a crystal can be represented as $\mathbf{M}=(\mathbf{A},\mathbf{X},\mathbf{L})$,
where $\mathbf{A}=[\boldsymbol{a}_1,\dots,\boldsymbol{a}_n]^\top\in\mathbb{R}^{n\times d_a}$
denotes the atomic features of $n$ atoms,
$\mathbf{X}=[\boldsymbol{x}_1,\dots,\boldsymbol{x}_n]^\top\in\mathbb{R}^{n\times3}$
contains their Cartesian coordinates,
and $\mathbf{L}=[\boldsymbol{l}_1,\boldsymbol{l}_2,\boldsymbol{l}_3]\in\mathbb{R}^{3\times3}$
is the lattice matrix composed of three basis vectors.
The entire crystal can then be expressed as
$(\hat{\mathbf{A}}, \hat{\mathbf{X}})=\{(\hat{\boldsymbol{a}_{i}}, \hat{\boldsymbol{x}}_i)|\hat{\boldsymbol{x}}_i=\boldsymbol{x}_i+k_1\boldsymbol{l}_1+k_2\boldsymbol{l}_2+k_3\boldsymbol{l}_3, \hat{\boldsymbol{a}_{i}}=\boldsymbol{a}_{i},k_1,k_2,k_3\in\mathbb{Z},i\in\mathbb{Z},1\leq i\leq n\}$,
which enumerates all atomic positions generated by lattice translations.

Alternatively, using the lattice matrix $\mathbf{L}$ as the basis vectors leads to the fractional coordinate representation,
where each atom is assigned a fractional coordinate
$\boldsymbol{f_{i}}=[f_{i,1},f_{i,2},f_{i,3}]^{\top}\in[0,1)^{3}$
and its Cartesian coordinate is given by
$\boldsymbol{x}_{i}=\sum_j f_{i,j}\boldsymbol{l}_{j}$.
Accordingly, the crystal can also be represented as
$\mathbf{M}=(\mathbf{A},\mathbf{F},\mathbf{L})$,
where $\mathbf{F}=[\boldsymbol{f}_1,\cdots,\boldsymbol{f}_n]^{\top}\in[0,1)^{n\times 3}$
denotes the fractional coordinates of all atoms in the unit cell.

\subsection{Problem statement}

\textbf{Crystal Tensor Prediction.} 
The crystal tensor properties prediction is a classic regression task. 
Its goal is to estimate the high-order tensor property denoted as $\mathcal{Y}$ from the raw crystal data represented as $\mathbf{M}=(\mathbf{A},\mathbf{F},\mathbf{L})$. Given that the crystal data $\mathbf{M}$ resides within the input space $\mathcal{V}$, and the tensor property $\mathcal{Y}$ belongs to the separate output space $\mathcal{W}$,
the objective of crystal tensor prediction is to find a function ${f_\theta }:\mathcal{V} \to \mathcal{W}$ that accurately maps input crystal data to the desired tensor property.
This is achieved by minimizing the discrepancy between the true property $\mathcal{Y}$ and the predicted property value $\mathcal{Y}_{pred}$.
Therefore, the crystal tensor property prediction task can be mathematically formulated as the following optimization problem:
\begin{equation}
\label{statement}
\mathop {\min }\limits_\theta  \sum\limits_{n = 1}^N {||{\cal Y}_{pred}^{(n)} - {\cal Y}^{(n)}|{|^2}} ,\quad {\cal Y}_{pred} = {f_\theta }(\mathbf{A},\mathbf{F},\mathbf{L}),
\end{equation}
where $f_\theta(\cdot)$ represents a tensor prediction model with learnable parameters $\theta$, and the superscript $n$ denotes individual samples in the dataset. In what follows, we will omit superscript $n$ for simplicity. 
As previously proposed in the literature~\citep{mao2024dielectric,wen2024equivariant,yanspace}, 
our objective is to estimate high-order tensor properties, including 2-order dielectric tensor (i.e., $\mathcal{Y}\buildrel \Delta \over =\boldsymbol{\varepsilon} \in \mathbb{R}^{3\times3}$), 3-order piezoelectric tensor (i.e., $\mathcal{Y}\buildrel \Delta \over =\mathbf{e} \in \mathbb{R}^{3\times3\times3}$), and 4-order elastic tensor (i.e., $\mathcal{Y}\buildrel \Delta \over = C \in \mathbb{R}^{3\times3\times3\times3}$), respectively.

\subsubsection{$\mathbf{O(3)}$ group.}The $O(3)$ group is an orthogonal group, consisting of rotations and reflections.
The elements $g \in O(3)$ act on vectors or tensors in $\mathcal{V}$, 
$\mathcal{W}$ through their respective group representations, $\rho_\mathcal{V}:O(3) \to GL(\mathcal{V})$ and $\rho_\mathcal{W} : O(3) \to GL(\mathcal{W})$, where $GL(\mathcal{V})$ and $GL(\mathcal{W})$ are the space of invertible linear maps $\mathcal{V} \to \mathcal{V}$ and $\mathcal{W} \to \mathcal{W}$, respectively. Specifically, the action of an element $g \in O(3)$ on crystal data $\mathbf{M}$ is defined as: $g\cdot\mathbf{M}=\rho_\mathcal{V}(g)\mathbf{M}=(\mathbf{A},\mathbf{F},\mathbf{Q}\mathbf{L})$.
For dielectric tensor property $\mathcal{Y}$, the action of $g$ is given by:
$g\cdot\mathcal{Y}=\rho_\mathcal{W}(g)\mathcal{Y}=\mathbf{Q}\mathcal{Y}\mathbf{Q}^\top$,
where $\mathbf{Q}$ is a $3\times3$ orthogonal matrix
(for additional details regarding the transformation of tensor properties, please refer to Appendix \ref{tensor_equivariance}).

\subsubsection{$\mathbf{O(3)}$-Equivariance for Crystal Tensor Prediction.}In the crystal tensor properties prediction task, due to the difference between $\rho_\mathcal{V}(g)$ and $\rho_\mathcal{W}(g)$, the requirements for $O(3)$ equivariance typically differ from the $O(3)$-equivariance defined in the general molecular studies \citep{hoogeboom2022equivariant,xuequivariant,zheng2024relaxing,song2024equivariant,aykent2025gotennet, cen2024high}. 
Specifically, taking the dielectric tensor as an example where $\mathcal{Y}_{label}\buildrel \Delta \over =\boldsymbol{\varepsilon} \in \mathbb{R}^{3\times3}$,
for a tensor prediction model $f_\theta(\cdot)$ in Eq.\ref{statement}, if it is $O(3)$ equivariant, it must satisfy the following equality formulated as: 
\begin{equation}
\label{equivariance_eq}
f_\theta(\mathbf{A},\mathbf{F},\mathbf{Q}\mathbf{L}) = \mathbf{Q}f_\theta(\mathbf{A},\mathbf{F},\mathbf{L})\mathbf{Q}^\top,
\end{equation}
where $\mathbf{Q}\in\mathbb{R}^{3 \times 3}$ is an arbitrary orthogonal matrix \citep{yanspace}. 
For clarity, an illustration of $O(3)$-equivariance for crystal tensor prediction is shown in
Fig.\ref{equvarience_tensor}
(for more equivariance details, see Appendix \ref{tensor_equivariance}).

\section{Methodology}
\label{Methodology}

In this section, we will first provide the motivation driving our new framework. 
We then reformulate the $O(3)$ equivariant tensor prediction task from the perspective of canonicalization.
Building upon this foundational perspective, we demonstrate how polar decomposition can be employed as a canonicalization strategy and effectively applied to solve tensor prediction tasks with both effectiveness and efficiency.

\subsection{The Motivation of Our Framework} 
\label{motivation}
To enforce $O(3)$ equivariance as defined in Eq. \ref{equivariance_eq}, existing crystal tensor prediction approaches typically rely on irreducible representations and Clebsch-Gordan (CG) tensor products. Although these formulations are theoretically rigorous and have shown strong empirical performance, they introduce substantial computational overhead due to the complexity of tensor operations.

Recently, several studies have explored the scalarization trick to exploit high-degree steerable features, thereby eliminating the need for CG tensor products while maintaining both accuracy and efficiency~\citep{aykent2025gotennet,cen2024high,han2025survey}. Similarly, if $O(3)$ tensor equivariance could be achieved for high-order tensor predictions without relying on tensor product operations, the computational cost could be significantly reduced.
Motivated by recent progress in canonicalization-based symmetric learning, we propose a principled framework that leverages canonicalization to realize $O(3)$-equivariant tensor prediction for crystalline materials. 

\subsection{$O(3)$ Equivariant Tensor Prediction from the Perspective of Canonicalization}

To mathematically describe $O(3)$-equivariant tensor prediction from the perspective of canonicalization, we first introduce and extend several fundamental concepts derived from canonicalization techniques \citep{ma2024canonicalization, kaba2023equivariance,dymequivariant}.

\begin{definition}[{\em Orbit.}]\label{definition_Orbit}
The orbit of a crystal $\mathbf{M}$ is defined as 
$\mathrm{Orbit}(\mathbf{M}) = \{ g \cdot \mathbf{M} \mid g \in O(3) \}$.
\end{definition}

In the context of crystal tensor prediction, all configurations of a crystal (i.e. a crystal wirh different orientations) in the space fall into the same orbit under the $O(3)$ group transformation. Specifically, for any two crystal configurations denoted as $\mathbf{M}_1$ and $\mathbf{M}_2$ within the same orbit, there exists a rotation or reflection $g \in O(3)$ such that $\mathbf{M}_1 = g \cdot \mathbf{M}_2$. Likewise, the tensor properties of the crystal also exhibit an orbit, i.e. $\text{Orbit}(\mathcal{Y}) = \{ g \cdot \mathcal{Y} \mid g \in O(3) \}$.
All configurations of a crystal within an orbit could be transformed into one another through group transformations. 
Consequently, we can select a specific configuration from the orbit to serve as a canonical representative (referred to as the \textit{canonical form}) that effectively characterizes the entire orbit. The formal process is outlined in the following definition.

\begin{definition}
\label{definition_Canonicalization}
{\em (Orbit Canonicalization\footnote{
For simplicity, we refer to the procedure as canonicalization, though it may more precisely be regarded as a \textit{quasi-canonicalization} that relaxes permutation invariance.
})}
A function $C_\mathbf{M}: \mathcal{V} \to \mathcal{V}$ is defined as an orbit canonicalization if there exists a canonical form $\mathbf{M}_0 \in \text{Orbit}(\mathbf{M})$ such that 
$C_\mathbf{M}(\mathbf{M}_1) = \mathbf{M}_0$ for all $\mathbf{M}_1 \in \text{Orbit}(\mathbf{M})$.
\end{definition}
Orbit canonicalization enables the transformation of a crystal with arbitrary orientations in the space into one with a specific, well-defined orientation.
Since the output of orbit canonicalization can always be mapped to the invariant canonical form, it follows that
$C_\mathbf{M}(\cdot)$ itself is $O(3)$-invariant.
Therefore, for any prediction function ${f_\theta }(\cdot)$, it holds that ${f_\theta }(C_\mathbf{M}(\cdot))$ remains $O(3)$-invariant (See the proof in Appendix \ref{Proofs}).

Although an
$O(3)$-invariant prediction function has been established via orbit canonicalization, the complete process of equivariant prediction requires specific group transformations that map all crystal configurations within the orbit to their canonical form. To address this requirement, we introduce the concept of rigid registration. Rigid registration is a fundamental concept in fields such as computer vision and computational geometry \citep{tam2012registration}. It aims to find a transformation that matches the position and orientation of one object with the corresponding position and orientation of another object. This concept can be naturally adapted to the crystal, as follows.

\begin{definition}
\label{definition_Rigid}
{\em (Rigid Registration.)}
A function $R_\mathbf{M}: \mathcal{V}\times\mathcal{V} \to O(3)$ is a rigid registration if $\forall\mathbf{M}_1,\mathbf{M}_2\in \text{Orbit}(\mathbf{M}),R_\mathbf{M}(\mathbf{M}_1,\mathbf{M}_2)=g$, such that $\mathbf{M}_1=g\cdot\mathbf{M}_2$.
\end{definition}
Using Definition \ref{definition_Rigid}, 
we can identify a transformation that matches two configurations of a crystal. Specifically, this allows us to find specific group transformations that map all crystal configurations within the orbit to
their canonical form.
By leveraging 
Definition \ref{definition_Canonicalization} and Definition \ref{definition_Rigid} together,
we derive a view for performing $O(3)$-equivariant tensor prediction through canonicalization, as described below.

\begin{proposition}
\label{proposition_Unified}
($O(3)$-Equivariant Tensor Prediction from the Perspective of Canonicalization.)
Given an arbitrary tensor prediction function 
$f(\mathbf{M}): \mathcal{V} \to \mathcal{W}$, 
we define a new function 
$h(\mathbf{M}) = R_\mathbf{M}\!\left(\mathbf{M}, C_\mathbf{M}(\mathbf{M})\right)
\!\cdot\! f\!\left(C_\mathbf{M}(\mathbf{M})\right)$,
such that $h(\mathbf{M})$ is $O(3)$-equivariant for tensor prediction. 
\end{proposition}

The proof of Proposition \ref{proposition_Unified} is provided in Appendix \ref{Proofs}. 
Specifically, in this unified view of $O(3)$ equivariant tensor prediction, 
we leverage rigid registration, as defined in Definition \ref{definition_Rigid}, to find specific group transformations that map crystal configurations within the orbit to
their canonical form. This canonical form is determined through orbit canonicalization, as described in Definition \ref{definition_Canonicalization}. Subsequently, we can directly apply these transformations to the network’s output, thereby effectively achieving the $O(3)$ equivariance. 
Under this perspective, our framework introduces no architectural constraints or special design requirements on the backbone network itself, since the $O(3)$ transformations of the crystal are externally determined through canonicalization and applied directly to the outputs.
Building upon Proposition \ref{proposition_Unified}, 
the key challenge in achieving $O(3)$-equivariant tensor prediction lies in identifying the appropriate orbit canonicalization and rigid registration functions.

\subsection{Our Proposed Framework: {\em GoeCTP}} 
\label{sec:goectp}

In what follows, we demonstrate that polar decomposition can effectively serve as both orbit canonicalization and rigid registration functions for crystal tensor prediction.
Building on this foundation, 
we will first introduce the core rotation and reflection (R\&R) module of the proposed {\em GoeCTP}, which is to obtain the canonical form for a crystal with arbitrary spatial orientations using polar decomposition. Then, we will describe how the input crystal data is processed and introduce the property prediction network of {\em GoeCTP}. Finally, we will explain how proposed {\em GoeCTP} achieves equivariant tensor predictions. An overview of the {\em GoeCTP} framework is illustrated in Fig.\ref{fig:overview}.

\begin{theorem}
 (Polar Decomposition \citep{hall2013lie,higham1986computing,jiaospace}.) An invertible matrix $\mathbf{L}\in\mathbb{R}^{3\times3}$ can be uniquely decomposed into $\mathbf{L}=\mathbf{Q}\mathbf{H}$, where $\mathbf{Q}\in\mathbb{R}^{3\times3}$ is an orthogonal matrix, $\mathbf{H}\in\mathbb{R}^{3\times3}$ is a symmetric positive semi-definite matrix. \label{proposition_polar}
\end{theorem}

In the fractional coordinate system, the $O(3)$ group transformations applied to a crystal affect only the lattice matrix $\mathbf{L}$, while fractional coordinates remain invariant, making it convenient for rotating the crystal configuration to the canonical form. Therefore, throughout this work, we adopt the fractional coordinate system to represent the crystal represented as $\mathbf{M} = (\mathbf{A}, \mathbf{F}, \mathbf{L})$.
Then building on Theorem 1,  we can define both orbit canonicalization and rigid
registration functions for crystal tensor prediction as follows.

\begin{proposition}
\label{proposition_plolar_Canonical}
(Orbit Canonicalization and Rigid Registration for Crystal Tensor Prediction.) 
By performing polar decomposition on the lattice matrix $\mathbf{L}=\mathbf{Q}\mathbf{H}$ of a crystal $\mathbf{M}$, 
we can define functions $f_{p1}(\mathbf{M})=(\mathbf{A}, \mathbf{F}, \mathbf{H})$ and $f_{p2}(\mathbf{M},f_{p1}(\mathbf{M}))=\mathbf{Q}$, where $f_{p1}(\cdot)$ and $f_{p2}(\cdot)$ correspond to performing polar decomposition on the lattice matrix. In this perspective,
$f_{p1}(\cdot)$ serves as the orbit canonicalization function, while $f_{p2}(\cdot)$ serves as the rigid registration function. The crystal configuration $\mathbf{M}_0=(\mathbf{A}, \mathbf{F}, \mathbf{H})$ is thus identified as the canonical form.
\end{proposition}

The detailed proof of Proposition \ref{proposition_plolar_Canonical} is provided in Appendix \ref{Proofs}.
Specifically, 
since the polar decomposition of a lattice matrix always exists and it is unique, any orthogonal matrix acting on the lattice matrix can be separated through polar decomposition, yielding a unique $\mathbf{H}$. Thus, $\mathbf{H}$ can serve as the canonical form.
Moreover, using polar decomposition is a continuous canonicalization, which can provide improved robustness \citep{dymequivariant}. For clarity, we provide details on this continuity in Appendix \ref{non_continuous}, and further discussions on constructing alternative canonical forms based on polar decomposition in Appendix \ref{local_frame_design}, respectively. Building upon Proposition \ref{proposition_plolar_Canonical}, we utilize polar decomposition to construct the proposed \textit{GoeCTP} model.

\subsubsection{R\&R Module.}The primary function of the R\&R Module is to rotate the crystal configuration to its canonical form and obtain the corresponding transformation. 
Therefore, 
based on Proposition \ref{proposition_plolar_Canonical},
the R\&R Module directly applies polar decomposition to the lattice matrix $\mathbf{L}$ to construct the rigid registration and the orbit canonicalization. Formally, the R\&R Module can be expressed as:
\begin{equation}
    f_{p1}(\mathbf{M})=(\mathbf{A},\mathbf{F},\mathbf{H}),\text{ }
    f_{p2}(\mathbf{M},f_{p1}(\mathbf{M}))=\mathbf{Q},
\end{equation}
where $f_{p1}(\cdot)$ is the orbit canonicalization function, while $f_{p2}(\cdot)$ is the rigid registration function.
As shown in Fig. \ref{fig:overview}, the crystal configuration $(\mathbf{A},\mathbf{F},\mathbf{H})$ derived from orbit canonicalization is passed to the subsequent crystal graph construction module for further processing. The orthogonal matrix $\mathbf{Q}$ obtained during this decomposition is passed to the reverse R\&R module, ensuring the equivariant transformation of the output tensor properties. 
Our proposed R\&R module based on polar decomposition allows the input crystal data to be transformed into the canonical form that is invariant under $O(3)$ space group transformations.
With this particular module, the equivariance can be captured, meaning that the subsequent components of {\em GoeCTP} are no longer required to account for equivariance.

\subsubsection{Crystal Graph Preprocessing.}
To enable neural networks to effectively process infinite crystal structures $(\mathbf{A}, \mathbf{F}, \mathbf{H})$, it is necessary to employ crystal graph construction and node and edge feature embedding techniques that represent atomic interactions within a finite graph representation.
In this work, we adopt the crystal graph construction approach introduced by \citet{yancomplete, yanspace}, where atoms are represented as nodes and interatomic interactions as edges. The node and edge features are subsequently embedded into vector representations $\boldsymbol{f}_{i}$ and $\boldsymbol{f}_{ij}^e$ to capture atomic and geometric information.
A detailed description of the crystal graph construction procedure and the feature embedding scheme is provided in Appendix~\ref{ecomformer}.

\subsubsection{Prediction Module.}Since the R\&R module is responsible for preserving equivariance, any predictive network can serve as the Prediction module, such as those proposed by \citet{yancomplete, taniaicrystalformer, yan2022periodic, lee2024density, wang2024conformal,ito2025rethinking}, among others. In this work, we select eComFormer~\citep{yancomplete} as the Prediction module to illustrate the overall workflow of our proposed method. A detailed explanation of eComFormer can be found in Appendix \ref{ecomformer}. Once processed through the stacked layers of eComFormer, the node features are aggregated to generate the crystal’s global features as follows:
$\boldsymbol{G}^\mathrm{final}=\frac1n\sum_{1\leq i\leq n}\boldsymbol{f}_i^\mathrm{final}$.

\subsubsection{Reverse R\&R Module.} The primary function of the Reverse R\&R module is to generate the equivariant tensor property predictions based on the crystal global features from the Prediction module. 
First, the Reverse R\&R module transforms $\boldsymbol{G}^\mathrm{final}$ into a tensor output, as follows:
\begin{equation}
\boldsymbol{\varepsilon}=f_{MLP}(\boldsymbol{G}^\mathrm{final})
\end{equation}
\begin{equation}
\boldsymbol{\varepsilon}^\mathrm{final}=f_{rp}(\boldsymbol{\varepsilon},\mathbf{Q}),
\end{equation}
where $f_{MLP}(\cdot)$ represents a multilayer perceptron (MLP) and the operation reshaping dimension. Next, $f_{rp}(\cdot)$ utilizes the orthogonal matrix $\mathbf{Q}$ obtained from the R\&R Module to convert the tensor output $\boldsymbol{\varepsilon}$ into its final equivariant form denoted as $\boldsymbol{\varepsilon}^\mathrm{final}$.
This conversion $f_{rp}(\cdot)$ for predicting the dielectric tensor can be expressed by: 
\begin{equation}
\label{equation11}
\boldsymbol{\varepsilon}^\mathrm{final}=\mathbf{Q}\boldsymbol{\varepsilon}\mathbf{Q}^\top. 
\end{equation}
For predicting the higher-order piezoelectric and elastic tensor, the conversion process becomes more complex, see Appendix \ref{tensor_equivariance} for more details.

\section{Experiments}
\label{Experiments}
\subsection{Experimental setup}

\begin{table}[b]
    \centering
\resizebox{0.48\textwidth}{!}{ 
    \begin{tabular}{c|c |c |c|c }
        \toprule

        Dataset &Sample size &Fnorm Mean  & Fnorm STD  &Unit \\       
        \midrule
        Dielectric& 4713 &14.7 & 18.2 & Unitless \\

        Piezoelectric&  2701 & 0.79 & 4.03 & $\text{C}/\text{m}^2$ \\
        Elastic& 25110  & 306.4 & 238.4 & GPa \\

        \bottomrule
    \end{tabular}
    }
    \caption{Dataset statistics.}
       \label{Dataset_statistics}
\end{table}

\subsubsection{Datasets.}In this work, we evaluate the performance of {\em GoeCTP} on three key tensor property prediction tasks: the second-order dielectric tensor, the third-order piezoelectric tensor, and the fourth-order elastic tensor, respectively. 
The datasets for dielectric tensor and piezoelectric tensor prediction are derived from the data curated by \citet{yanspace}, which originates from the JARVIS-DFT database.
Since nearly half of the crystal samples in the piezoelectric dataset have piezoelectric tensor labels that are entirely zero, training {\em GoeCTP} directly on this dataset can lead to severe overfitting.
Therefore, we filtered out all crystal samples with zero-valued piezoelectric tensor labels from piezoelectric dataset.
For the elastic tensor prediction task, the dataset is also obtained from the JARVIS-DFT database and is publicly accessible as \texttt{dft\_3d} through the \texttt{jarvis-tools} package \footnote{\url{https://pages.nist.gov/jarvis/databases/}}.
The statistical details of the datasets are presented in Table \ref{Dataset_statistics}.
Further details of the dataset and experimental setup can be found in Appendix \ref{Experimental_details}.

\subsubsection{Baseline Methods.}We selected several state-of-the-art approaches in the field of crystal tensor property prediction, including EGTNN \citep{zhong2023general} and GMTNet \citep{yanspace}, as baseline methods.
Moreover, since {\em GoeCTP} is designed as a flexible framework, we further evaluated its performance when combined with different backbone models. In addition to the eComFormer~\citep{yancomplete} used in Figure \ref{fig:overview}, we also incorporated the recently proposed CrystalFramer~\citep{ito2025rethinking}.

\subsubsection{Evaluation Metrics.}We followed the evaluation metrics defined by \citet{yanspace} to assess the performance of the methods. The following metrics were employed: (1) \textbf{Frobenius norm (Fnorm)} is used to measure the difference between the predicted tensor and the label tensor, which is the square root of the sum of the squares of all elements in a tensor. 
(2) \textbf{Error within threshold (EwT)} is determined by the ratio of the Fnorm between the predicted tensor and the ground truth tensor to the Fnorm of the ground truth tensor. This ratio can be expressed as ${{||{y_{pred}} - {y_{label}}|{|_F}} \mathord{\left/
 {\vphantom {{||{y_{pred}} - {y_{label}}|{|_F}} {||{y_{label}}|{|_F}}}} \right.
 \kern-\nulldelimiterspace} {||{y_{label}}|{|_F}}}$, where ${||\cdot|{|_F}}$ is Fnorm, and $y_{label}$ and $y_{pred}$ represent the ground truth and predicted values, respectively. For instance, EwT 25\% indicates that the proportion of predicted samples with this ratio is below 25\%. Higher values of EwT signify better prediction quality.
In our experiments, we utilized several thresholds for EwT: EwT 25\%, EwT 10\%, and EwT 5\%.

\subsection{Experimental Results}

\subsubsection{Predicting Dielectric Tensors.}
The performance of various models in predicting the dielectric tensor is summarized in Table \ref{tab:main_results}.  
Since {\em GoeCTP (C. Fra.)} employs a more advanced backbone, its performance is notably superior to that of {\em GoeCTP (C. eCom.)}. Both {\em GoeCTP (eCom.)} and {\em GoeCTP (C. Fra.)} outperform the baseline methods, demonstrating the effectiveness of the proposed {\em GoeCTP} framework.
Comprehensive results combining {\em GoeCTP} with additional backbone architectures are provided in Appendix \ref{5_times}. 
In addition to these experimental results, we further investigate the impact of different canonical forms on tensor property prediction in crystalline materials. Detailed analyses can be found in Appendix \ref{local_frame}. 
\begin{table}[h]
    \centering
        \resizebox{0.48\textwidth}{!}{ 
\begin{tabular}{l|cc|cc}
\toprule
 & ETGNN & GMTNet & \textbf{GoeCTP (eCom.)}& \textbf{GoeCTP (C. Fra.)} \\
\midrule
Fnorm $\downarrow$         & 3.40 & 3.28 & 3.23& \textbf{3.05} \\
EwT 25\% $\uparrow$         & 82.6\% & 83.3\% & 83.2\% & \textbf{86.4\%}\\
EwT 10\% $\uparrow$       & 49.1\% & 56.0\% & 56.8\% & \textbf{62.6\%}\\
EwT 5\% $\uparrow$         & 25.3\% & 30.5\% & 35.5\%& \textbf{43.5\%}\\
\bottomrule
\end{tabular}
}
\caption{Comparison on the dielectric dataset. 
Lower Fnorm and higher EwT indicate better performance. {\em GoeCTP (eCom.)} denotes GoeCTP using eComFormer as the backbone, while {\em GoeCTP (C. Fra.)} refers to {\em GoeCTP} with CrystalFramer as the backbone.}
\label{tab:main_results}
\end{table}

\subsubsection{Predicting Piezoelectric Tensors.}The experimental results for the piezoelectric tensor dataset are shown in Table \ref{tab:piezoelectric_results}. 

Compared to the results on the dielectric dataset, predicting the piezoelectric tensor presents a greater challenge.
As shown in Table \ref{Dataset_statistics}, the overall mean Fnorm of the piezoelectric tensor is relatively small, which leads to lower overall EwT values. Consequently, achieving higher EwT requires more accurate predictions from the model.
{\em GoeCTP (C. Fra.)} achieves relatively high EwT at both 5\% and 10\%, indicating that it attains high prediction accuracy for a subset of samples.
Meanwhile, GMTNet and {\em GoeCTP (eCom.)} obtain better Fnorm results, suggesting that their overall prediction errors are comparatively smaller.

\begin{table}[h]
    \centering
        \resizebox{0.48\textwidth}{!}{ 
\begin{tabular}{l|cc|cc}
\toprule
 & ETGNN & GMTNet & \textbf{GoeCTP (eCom.)}& \textbf{GoeCTP (C. Fra.)} \\
\midrule
Fnorm $\downarrow$   & 0.873 &\textbf{0.752}  & 0.778 & 0.827\\
EwT 25\% $\uparrow$  & 0\% &\textbf{ 6.29\%} & 2.59\% & 3.33\%\\
EwT 10\% $\uparrow$   & 0\% & 1.48\% & 1.14\% & \textbf{2.59\%}\\
EwT 5\% $\uparrow$    & 0\% & 1.11\% & 0.04\% & \textbf{2.22\%}\\
\bottomrule
\end{tabular}
}
\caption{Comparison on the piezoelectric dataset. }
\label{tab:piezoelectric_results}
\end{table}

\subsubsection{Predicting Elastic Tensors.}Experimental results on the elastic tensor dataset are summarized in Table \ref{tab:elastic_results}. Consistent with the observations on the dielectric dataset, both {\em GoeCTP (eCom.)} and {\em GoeCTP (C. Fra.)} achieved outstanding performance in predicting higher-order, complex tensors, surpassing all baseline methods across all evaluation metrics. Notably, {\em GoeCTP (C. Fra.)} delivered a further improvement over {\em GoeCTP (eCom.)}. Additional results demonstrating {\em GoeCTP} combined with other backbone models are provided in Appendix \ref{5_times}.

\begin{table}[h]
    \centering
        \resizebox{0.48\textwidth}{!}{ 
\begin{tabular}{l|cc|cc}
\toprule
 & ETGNN & GMTNet & \textbf{GoeCTP (eCom.)}& \textbf{GoeCTP (C. Fra.)} \\
\midrule
Fnorm $\downarrow$    & 123.64 & 117.62 & 107.11&  \textbf{95.98}\\
EwT 25\% $\uparrow$   & 32.0\% & 36.0\% & 42.5\%& \textbf{49.7\%} \\
EwT 10\% $\uparrow$    & 3.8\%  & 7.6\%  & 15.3\%& \textbf{20.3\%}\\
EwT 5\% $\uparrow$   & 0.5\%  & 2.0\%  & 7.2\%& \textbf{9.8\%}\\
\bottomrule
\end{tabular}
}
\caption{Comparison on the elastic dataset. }
\label{tab:elastic_results}
\end{table}

\begin{table*}[t]
    \centering
        \resizebox{0.7\textwidth}{!}{ 
\begin{tabular}{lcccc}
\toprule
& eCom. (ori.) & eCom. (aug.) & \textbf{GoeCTP (ori.)} & \textbf{GoeCTP (aug.)} \\
\midrule
Dielectric: Fnorm $\downarrow$ & 3.23 & 4.71 & 3.23 & 3.23 \\
Dielectric: Test Time (s) $\downarrow$ & 26.03 & 26.01 & 26.23 & 26.18 \\
\midrule
Piezoelectric: Fnorm $\downarrow$ & 0.778 & 0.957 & 0.778 & 0.778 \\
Piezoelectric: Test Time (s) $\downarrow$ & 5.21 & 5.40 &5.53  & 6.05 \\
\midrule
Elastic: Fnorm $\downarrow$ & 107.11 & 138.45 & 107.11 & 107.11 \\
Elastic: Test Time (s) $\downarrow$ & 83.26 & 83.02 & 90.10 & 89.60 \\
\bottomrule
\end{tabular}
}
\caption{Ablation study for verifying the $O(3)$-equivariance using dielectric, piezoelectric, and elastic datasets. After training {\em GoeCTP (eCom.)}, we extracted its prediction module (i.e., eComFormer) for comparative testing on two different test sets.}
\label{Ablation_table}
\end{table*}

\subsubsection{Verifying the $O(3)$ Equivariance.}
To evaluate the effectiveness of {\em GoeCTP}, we conducted experiments to verify the $O(3)$ equivariance of tensor properties. Specifically, after training {\em GoeCTP (eCom.)}, we extracted its Prediction module (i.e., eComFormer) for comparative testing on two different test sets (original test set and augmented test set). All crystals in the original test set were adjusted to the canonical form, while the augmented test set was generated by applying arbitrary $O(3)$ group transformations to all crystals in the original test set. The method for generating the corresponding orthogonal matrices is from \citet{heiberger1978algorithm}.
We then evaluated both the Prediction module (eComFormer) and {\em GoeCTP (eCom.)} on these two datasets and compared the performance metrics.

The results on the three tensor datasets are shown in 
Table
\ref{Ablation_table}, 
{\em GoeCTP} performed equally well on the augmented test set as on the original test set in all tensor datasets, indicating that it maintains strong $O(3)$ equivariance for tensor properties. In contrast, the performance of the eComFormer method significantly declined on the augmented test set, demonstrating that it does not meet the $O(3)$ equivariance requirements for tensor properties. 
Additionally, runtime comparisons between {\em GoeCTP} and eComFormer on the test set revealed that {\em GoeCTP} did not result in a significant increase in runtime. This demonstrates that when integrated into scalar property prediction networks, {\em GoeCTP} incurs almost no additional computational cost,  enhancing its practicality while maintaining efficiency.

\begin{table}[h]
    \centering
        \resizebox{0.48\textwidth}{!}{ 
\begin{tabular}{l|cc|cc}

\bottomrule

\toprule
Dielectric:& ETGNN & GMTNet & \textbf{GoeCTP (eCom.)}& \textbf{GoeCTP (C. Fra.)}\\
\midrule
Total Time (s) $\downarrow$  & 1325 & 1611 & \textbf{616}& 1976\\
Time/epoch (s) $\downarrow$  & 6.625 & 8.055 & \textbf{3.08}& 9.88\\
\bottomrule

\toprule
Piezoelectric:& ETGNN & GMTNet & \textbf{GoeCTP (eCom.)}& \textbf{GoeCTP (C. Fra.)}\\
\midrule
Total Time (s) $\downarrow$  & 736 & 7012 & \textbf{356} & 5768\\
Time/epoch (s) $\downarrow$  & 2.18 & 35.06 & \textbf{1.78}& 5.77\\
\bottomrule

\toprule
Elastic:&  ETGNN & GMTNet & \textbf{GoeCTP (eCom.)}& \textbf{GoeCTP (C. Fra.)}\\
\midrule
Total Time (s) $\downarrow$  & 4448 & $>$\,36000 & \textbf{2422}& 17197\\
Time/epoch (s) $\downarrow$  & 22.24 & $>$\,180  & \textbf{12.11}& 34.394\\

\bottomrule

\toprule
\end{tabular}
}
\caption{Efficiency comparison.}
\label{tab:efficiency}
\end{table}

\subsubsection{Efficiency.}Table \ref{tab:efficiency} summarizes the training efficiency of {\em GoeCTP} and baseline methods. On the dielectric, piezoelectric, and elastic datasets, {\em GoeCTP (eCom.)} required only 38.2\%, 5\%, and 7\% of GMTNet’s training time, respectively. Although {\em GoeCTP (C. Fra.)} achieved higher accuracy at the cost of longer training, it still maintained a efficiency advantage over GMTNet, particularly on high-order tensor datasets such as the elastic dataset.

To achieve $O(3)$ equivariance for tensor properties, GMTNet employs a network architecture based on irreducible representations and tensor operations. Moreover, its final tensor output is obtained through gradient computation, which may significantly reduce computational efficiency. This inefficiency becomes particularly pronounced when computing gradients for high-order tensors such as elastic tensors. Beyond training costs, the gradient computation also may lead to slower inference. For instance, in the elastic tensor task, GMTNet requires more than 0.478 seconds per material, whereas {\em GoeCTP (eCom.)} requires only 0.036 seconds per material. In contrast, {\em GoeCTP} achieves equivariant tensor predictions without any specialized architectural constraints. It utilizes a simple multilayer perceptron (MLP) at the network output to predict tensor properties of various orders, ensuring high efficiency and scalability across different tensor orders.

\section{Conclusion}

In this work, we revisit crystal tensor property prediction through the lens of canonicalization. Built upon the notion of canonicalization, we proposed a novel $O(3)$-equivariant framework {\em GoeCTP} for fast and accurate crystal tensor prediction. Benefiting from the canonicalization mechanism, {\em GoeCTP} serves as a plug-and-play module that can be integrated with any existing scalar property prediction network, enabling it to predict tensor properties with negligible additional computational cost.

Except for the excel performance, our current {\em GoeCTP} has some limitations remain: 
\begin{enumerate}
    \item {\em GoeCTP} is a plug-and-play $O(3)$-equivariant framework designed to enhance the backbone network’s capability for equivariant prediction. Consequently, its overall performance inherently depends on the strength of the chosen backbone. For example, {\em GoeCTP (eCom.)} demonstrates higher efficiency, while {\em GoeCTP (C. Fra.)} achieves greater predictive accuracy. Striking an optimal balance between efficiency and accuracy remains a topic for further investigation. 
    \item As demonstrated in \citet{higham1986computing} and Theorem \ref{proposition_polar}, polar decomposition applied to an invertible matrix $\mathbf{L}$ produces a unique canonical form $\mathbf{H}$. For 3D crystal structures, the lattice matrix $\mathbf{L}$ is always full rank (i.e., invertible), ensuring the applicability of our approach to 3D crystal systems. However, in special cases such as 2D single-layer crystals \citep{novoselov2005two, sherrell2022bright}, the lattice matrix $\mathbf{L}$ may be rank-deficient. In these scenarios, polar decomposition may fail to yield a unique canonical form $\mathbf{H}$, making it impossible to consistently align crystals with different spatial orientations to a common canonical form, thereby hindering $O(3)$-equivariance.
    \item The current framework underutilizes prior knowledge related to space group constraints and does not yet ensure that predicted tensor results strictly adhere to space group symmetry constraints. In future work, we plan to improve model performance by integrating priors concerning the independent tensor components across different crystal systems (see Appendix \ref{Voigt}). Further discussion on incorporating space group symmetry constraints into tensor property prediction is provided in Appendix \ref{symmetry_utilize}.
\end{enumerate}


\section{Acknowledgments}
This work was partially supported by the Research Grants Counil (RGC) of the Hong Kong (HK) SAR (Grant No. 15208725 and 15208222), the Young Scientists Fund of National Natural Science Foundation of China (NSFC) (Grant No. 62206235), and the Hong Kong Polytechnic University (Grant No. A0046682 and P0057774).

\bibliography{aaai2026}

\newpage

\appendix
\onecolumn
\setcounter{secnumdepth}{2}
\section{Appendix }
\subsection{Related work}
\label{work}

\textbf{GNN-Based Methods.}
CGCNN is a pioneering GNN model specifically designed for handling crystal structures \citep{xie2018crystal}. This model proposed to represent crystal structures as multi-edge crystal graphs. It was applied to predict various scalar properties such as formation energy and band gap. Since then, several GNN methods have been developed to improve upon CGCNN through exploring various network designs or leveraging prior knowledge \citep{chen2019graph,louis2020graph,choudhary2021atomistic,das2023crysgnn,lin2023efficient}.
These GNN methods are primarily designed for scalar property prediction and do not address the prediction of high-order tensor properties, such as dielectric or elastic tensors. Furthermore, they lack the ability to preserve the equivariance required for accurate high-order tensor property prediction.
In contrast, recent studies attempted to ensure equivariance through specialized network architectures~\citep{mao2024dielectric,lou2024discovery,heilman2024equivariant,wen2024equivariant,pakornchote2023straintensornet,yanspace, zhong2023general}. 
These approaches generally employ harmonic decomposition to achieve equivariance for tensor properties. 
Within these network architectures, many operations are required, such as tensor products and combining irreducible representations. These processes significantly increase computational costs, especially when handling higher-order data.

\textbf{Transformer-Based Methods.}
Transformers, with their self-attention mechanism and parallel processing capabilities, are particularly well-suited for predicting crystal material properties. Matformer \citep{yan2022periodic}, one of the earliest Transformer-based networks used for crystal material property prediction, encoded crystal periodic patterns by using the geometric distances between the same atoms in neighboring unit cells. This addressed the issue in earlier GNN-based methods, including CGCNN, MEGNet, GATGNN, and others, which neglected the periodic patterns of infinite crystal structures. Subsequently, more advanced Transformer-based approaches were proposed. Some methods, such as ComFormer \citep{yancomplete}, CrystalFormer \citep{wang2024conformal}, CrystalFormer \citep{taniaicrystalformer}, and CrystalFramer  \citep{ito2025rethinking} typically incorporate either enhanced graph construction techniques or physical priors, exhibiting impressive results in the prediction of scalar properties. 
Furthermore, DOSTransformer \citep{lee2024density} is tailored for the density of states prediction, utilizing prompt-guided multi-modal transformer architecture to achieve super performance.
Nevertheless, these models are not directly applicable for accurate high-order tensor property prediction due to their inability to capture the necessary equivariance. 

\textbf{Equivariant Learning.} 
Equivariance can be achieved through specialized architectural designs, such as incorporating high-degree steerable features \citep{han2025survey}. However, due to the structural constraints imposed by equivariant architectures, these models are often difficult to optimize \citep{manolache2025learning,xie2025tale}, and their complex designs may limit the universal approximation property \citep{pacini2025universality,cenuniversally}.
Recently, two general-purpose techniques for achieving equivariant learning under arbitrary groups have gained widespread attention: frame averaging \citep{linequivariance, punyframe, duval2023faenet} and canonicalization \citep{ma2024canonicalization, kaba2023equivariance, dymequivariant}. These methods represent promising alternatives to address the limitations of conventional equivariant models, as they impose no architectural constraints.
Frame averaging is developed based on group averaging and involves averaging over smaller frames instead of the entire group. 
Canonicalization, on the other hand, is inherently related to frame averaging and can be viewed as a more general form of it. \citep{ma2024canonicalization}. 
In particular, Lin et al. \citep{linequivariance} introduced minimal frame averaging, a method designed to achieve both efficient and exact equivalence in equivariant learning. 
This approach is primarily applicable to scenarios where the input and output spaces are identical, such as the n-body problem. It does not extend to cases involving different spaces, such as crystal tensor prediction.
If this method were extended to tensor property prediction, it would represent a special case of Proposition \ref{proposition_Unified}, where both orbit canonicalization and rigid registration are achieved using QR decomposition.
Within Proposition \ref{proposition_Unified}. However, different functions can be employed for orbit canonicalization and rigid registration. For instance, the Kabsch algorithm \citep{kabsch1976solution, lawrence2019purely} can be employed for rigid registration, while polar decomposition can be utilized for canonicalization.

To the best of our knowledge, no previous work has extended these equivariant learning techniques to the prediction of crystal tensor properties. Thus, our method represents the first extension of these concepts to the crystal tensor property prediction.

\color{black}

\subsection{Details of Framework Architecture}
\label{ecomformer}

{\bf Crystal Graph Construction.} 
In \textit{GoeCTP}, we use the crystal graph construction from \citet{yancomplete,yanspace} to describe the structure and relationships within crystals.
Specifically, we first convert the fractional coordinate system $(\mathbf{A}, \mathbf{F}, \mathbf{H})$ into the Cartesian coordinate system $(\mathbf{A}, \mathbf{X}, \mathbf{H})$ as introduced in Section 2.
Then, we assume that the output crystal graph is represented as $\mathcal{G}(\mathcal{V}, \mathcal{E})$, $\mathcal{V}$ denotes the set of nodes $v_i$ in the crystal graph, where each node $v_i$ contains atomic features $\mathbf{v}_i=(\boldsymbol{a}_{i}, \hat{\boldsymbol{p}}_i)$. $\mathcal{E}$ represents the set of edges denoted as $e_{ij}$, which are typically constructed based on the Euclidean distance $d_{ij}$ between nodes $v_i$ and $v_j$. 
When the Euclidean distance $d_{ij}$ between nodes $v_i$ and $v_j$ is less than a given radius $R$, i.e. $d_{ij}=||\boldsymbol{p}_{j} + k_{1}\boldsymbol{l}_{1} + k_{2}\boldsymbol{l}_{2} + k_{3}\boldsymbol{l}_{3} - \boldsymbol{p}_{i}||_{2} \leq R$,
an edge $e_{ij}$ will be built with edge feature $\mathbf{e}_{ij}=\boldsymbol{p}_{j} + k_{1}\boldsymbol{l}_{1} + k_{2}\boldsymbol{l}_{2} + k_{3}\boldsymbol{l}_{3} - \boldsymbol{p}_{i}$.
Here, $R$ is established based on the distance to the $k$-th nearest neighbor, and
different values of $\boldsymbol{k} = [k_{1}, k_{2}, k_{3}] \in \mathbb{Z}^{3}$ represent different edges between nodes $v_i$ and $v_j$. 
Upon construction, the edge features are denoted as $e_{ij}$,
while the node features, representing atomic properties, are denoted as $\boldsymbol{a}_{i}$.
Since there are no equivariance requirements for subsequent models, any other graph construction methods can be used to replace this part, such as \citet{wang2024conformal,yancomplete}.

{\bf Node and Edge Feature Embedding.} Building on previous work \citep{xie2018crystal,yanspace,yancomplete}, node features $\boldsymbol{a}_{i}$ are embedded into a 92-dimensional CGCNN feature vector $\boldsymbol{f}_{i}$. Edge features $\mathbf{e}_{ij}$ are decomposed into their magnitude $||\mathbf{e}_{ij}||_2$ and a normalized direction vector $\hat{\mathbf{e}}_{ij}$. The magnitude is further mapped to a term similar to potential energy, $-c/||\mathbf{e}_{ij}||_2$, through the application of a radial basis function (RBF) kernel for encoding \citep{lin2023efficient}. Subsequently, $\mathbf{e}_{ij}$ are embedded into feature vector $\boldsymbol{f}_{ij}^e$.

{\bf eComFormer.}
The eComFormer has demonstrated strong performance across a range of scalar property prediction tasks \citep{yancomplete}.
It is built on an SO(3)-equivariant crystal graph representation, where the interatomic distance vectors are employed to represent the edge features of the graph. The model converts this crystal graph into embeddings and utilizes a transformer architecture, incorporating both a node-wise transformer layer and a node-wise equivariant updating layer, to extract rich geometric information during the message-passing process.

\begin{figure}[h] 
  \centering   
  \includegraphics[width=\textwidth]{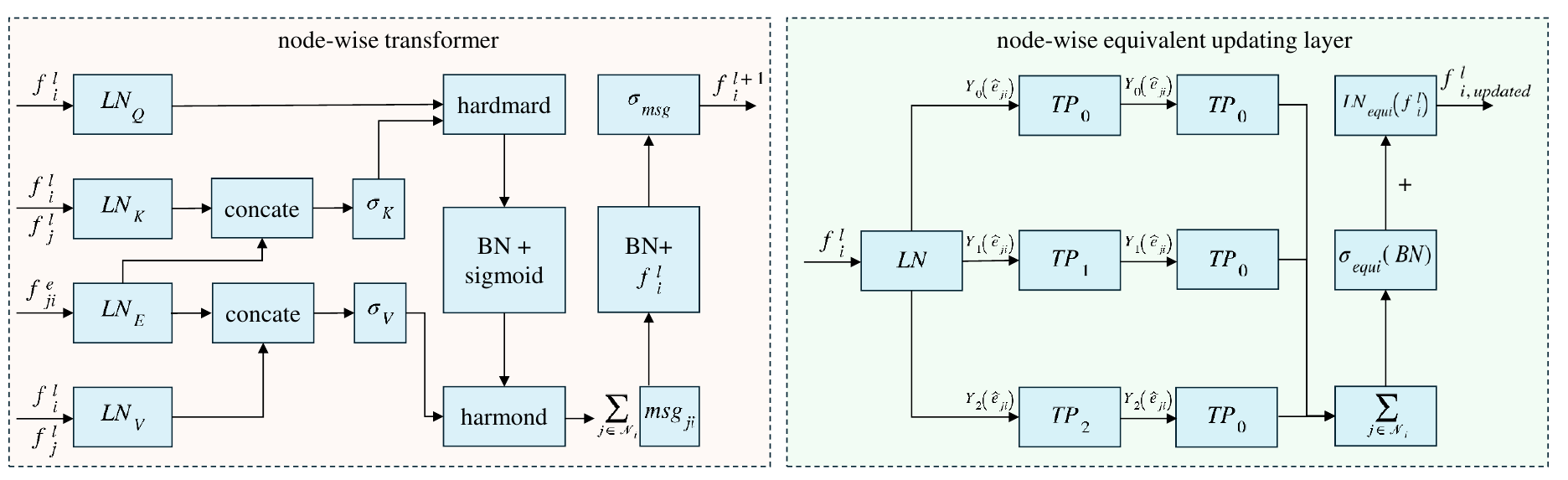}
  \caption{The detailed architectures of the node-wise transformer
layer and node-wise equivariant updating layer, adapted from \citet{yancomplete}.} 
  \label{fig:eComFormer}
\end{figure}

Specifically, the node-wise transformer layer is responsible for updating the node-invariant features $\boldsymbol{f}_{i}$. This process utilizes the node features $\boldsymbol{f}_{i}$, neighboring node features$\boldsymbol{f}_{j}$, and edge features $\boldsymbol{f}_{ij}$ to facilitate message passing from neighboring node $j$ to the central node $i$, followed by aggregation of all neighboring messages to update $\boldsymbol{f}_{i}$. The update mechanism is structured similarly to a transformer. Fristly, the message from node $j$ to node $i$ is transformed into corresponding query $\boldsymbol{q}_{ij}=\mathrm{LN}_Q(\boldsymbol{f}_i)$, key $\boldsymbol{k}_{ij}=(\mathrm{LN}_K(\boldsymbol{f}_i)|\mathrm{LN}_K(\boldsymbol{f}_j))$, and value feature $\boldsymbol{v}_{ij}=(\mathrm{LN}_V(\boldsymbol{f}_i)|\mathrm{LN}_V(\boldsymbol{f}_j)|\mathrm{LN}_E(\boldsymbol{f}_{ij}^e))$, where $\mathrm{LN}_Q(\cdot)$, $\mathrm{LN}_K(\cdot)$, $\mathrm{LN}_V(\cdot)$, $\mathrm{LN}_E(\cdot)$ denote the linear transformations, and $|$ denote the concatenation. Then, the self-attention
output is then computed as:
\begin{equation}
\boldsymbol{\alpha}_{ij}= \frac{\boldsymbol{q}_{ij}\circ\boldsymbol{\xi}_K(\boldsymbol{k}_{ij})}{\sqrt{d_{\boldsymbol{q}_{ij}}}},
\boldsymbol{msg}_{ij}=\operatorname{sigmoid}(\operatorname{BN}(\boldsymbol{\alpha}_{ij}))\circ\boldsymbol{\xi}_V(\boldsymbol{\upsilon}_{ij}), 
\end{equation}
where $\boldsymbol{\xi}_K$, $\boldsymbol{\xi}_V$ represent nonlinear transformations applied
to key and value features, respectively. The operators $\circ$ denote
the Hadamard product, $\operatorname{BN}(\cdot)$ refers to the batch normalization layer, and $\sqrt{d_{\boldsymbol{q}_{ij}}}$ indicates the dimensionality of $\boldsymbol{q}_{ij}$. Then, node feature $\boldsymbol{f}_{i}$ is updated as follows,
\begin{equation}
\boldsymbol{m}s\boldsymbol{g}_i=\sum_{j\in\mathcal{N}_i}\boldsymbol{m}s\boldsymbol{g}_{ij}, \boldsymbol{f}_i^\mathrm{new}=\boldsymbol{\xi}_{msg}(\boldsymbol{f}_i+\mathrm{BN}(\boldsymbol{m}s\boldsymbol{g}_i)),
\end{equation}
where $\boldsymbol{\xi}_{msg}(\cdot)$ denoting the softplus activation function.

The node-wise equivariant updating layer is designed to effectively capture geometric features by incorporating node feature $\mathbf{a}_i$ and edge feature $||\mathbf{e}_{ji}||_2$ as inputs and stacking two tensor product (TP) layers \citep{geiger2022e3nn}. It uses node feature $\boldsymbol{f}_{i}^l$ and equivalent vector feature $\mathbf{e}_{ji}$ embedded by corresponding spherical harmonics $\mathbf{Y}_0(\hat{\mathbf{e}}_{ji})=c_0$,$\mathbf{Y}_1(\hat{\mathbf{e}}_{ji})=c_1*\frac{\mathbf{e}_{ji}}{||\mathbf{e}_{ji}||_2}\in\mathbb{R}^3$ and $\mathbf{Y}_2(\hat{\mathbf{e}}_{ji})\in\mathbb{R}^5$ to represent the input features. Gathering rotational information from neighboring nodes to the central node $i$, the first TP layer is shown as 
\begin{equation}
\boldsymbol{f}_{i,0}^l=\boldsymbol{f}_i^{l^{\prime}}+\frac1{|\mathcal{N}_i|}\sum_{j\in\mathcal{N}_i}\mathbf{TP}_0(\boldsymbol{f}_j^{l^{\prime}},\mathbf{Y}_0(\hat{\mathbf{e}}_{ji})),\boldsymbol{f}_{i,\lambda}^l=\frac1{|\mathcal{N}_i|}\sum_{j\in\mathcal{N}_i}\mathbf{TP}_\lambda(\boldsymbol{f}_j^{l^{\prime}},\mathbf{Y}_\lambda(\hat{\mathbf{e}}_{ji})),\lambda\in\{1,2\},
\end{equation}
where $\boldsymbol{f}_i^{l^{\prime}}$ is linearly transformed from $\boldsymbol{f}_{i}^l$, $|\mathcal{N}_i|$ and $\mathbf{TP}_\lambda$ represent the number of neighbors of node $i$ and TP layers with rotation order $\lambda$ respectively. Then, to represent the invariant node features, the second TP layer is written as 
\begin{equation}    \boldsymbol{f}_i^{l*}=\frac1{|\mathcal{N}_i|}(\sum_{j\in\mathcal{N}_i}\mathbf{TP}_0(\boldsymbol{f}_{j,0}^l,\mathbf{Y}_0(\hat{\mathbf{e}}_{ji}))+\sum_{j\in\mathcal{N}_i}\mathbf{TP}_0(\boldsymbol{f}_{j,1}^l,\mathbf{Y}_1(\hat{\mathbf{e}}_{ji}))+\sum_{j\in\mathcal{N}_i}\mathbf{TP}_0(\boldsymbol{f}_{j,2}^l,\mathbf{Y}_2(\hat{\mathbf{e}}_{ji}))),
\end{equation}
stacking the two tensor product layers together using both linear and nonlinear transformations, the output $\boldsymbol{f}_{i,updated}^l$ is combined as
\begin{equation}
\boldsymbol{f}_{i,updated}^l=\boldsymbol{\sigma}_{\mathrm{equi}}(\mathrm{BN}(\boldsymbol{f}_i^{l*}))+\mathrm{LN}_{\mathrm{equi}}(\boldsymbol{f}_i^l),
\end{equation}
with $\boldsymbol{\sigma}_{\mathrm{equi}}$ denoting a nonlinear transformation made up of two softplus layers with a linear layer positioned between them.
The detailed architectures of the node-wise transformer layer and node-wise equivariant
updating layer are shown in Fig. \ref{fig:eComFormer}.

\subsection{Proofs}
\label{Proofs}

\subsubsection{Proofs of Lemma \ref{lemma_invariant}}

\begin{lemma}
\label{lemma_invariant}
Given a orbit canonicalization function $C_\mathbf{M}(\cdot)$,
for any prediction function ${f_\theta }(\cdot)$, it holds that ${f_\theta }(C_\mathbf{M}(\cdot))$ is $O(3)$-invariant.

\end{lemma}

\begin{proof}
According to Definition 2, for any $g \in O(3)$ and $\mathbf{M} \in \mathcal{V}$, we have $C_\mathbf{M}(\mathbf{M})=\mathbf{M_0}$ and $C_\mathbf{M}(g\cdot\mathbf{M})=\mathbf{M_0}$. Thus, 
\begin{equation}
\begin{aligned}
{f_\theta }(C_\mathbf{M}(g\cdot\mathbf{M}))&={f_\theta }(\mathbf{M_0})\\&={f_\theta }(C_\mathbf{M}(\mathbf{M})),     
\end{aligned}     
\end{equation}
which implies ${f_\theta }(C_\mathbf{M}(\cdot))$ is $O(3)$-invariant.
\end{proof}

\subsubsection{Proofs of Lemma \ref{lemma_equivariant}}

\begin{lemma}
\label{lemma_equivariant}

For a rigid registration function $R_\mathbf{M}(\mathbf{M}_1,\mathbf{M}_2)=g$,
it is $O(3)$-equivariant with respect to the input $\mathbf{M}_1$.

\end{lemma}

\begin{proof}
According to Definition 3, $\mathbf{M}_1=R_\mathbf{M}(\mathbf{M}_1,\mathbf{M}_2)\cdot\mathbf{M}_2$. Thus, for any $g \in O(3)$, we have 
\begin{equation}
\label{eq15}
\begin{aligned}
R_\mathbf{M}(g\cdot\mathbf{M}_1,\mathbf{M}_2)\cdot\mathbf{M}_2 &=g\cdot\mathbf{M}_1\\
&=g\cdot(R_\mathbf{M}(\mathbf{M}_1,\mathbf{M}_2)\cdot\mathbf{M}_2)\\
&=(g\cdot R_\mathbf{M}(\mathbf{M}_1,\mathbf{M}_2))\cdot\mathbf{M}_2.
\end{aligned}     
\end{equation}
Given $\mathbf{M}_2=(\mathbf{A}_2,\mathbf{F}_2,\mathbf{L}_2)$,
and noting that $g\cdot\mathbf{M}_2=(\mathbf{A}_2,\mathbf{F}_2,\mathbf{Q}\mathbf{L}_2)$ (where $\mathbf{L}_2$ has full rank), we can write Eq. \ref{eq15} as
\begin{equation}
(R_\mathbf{M}(g\cdot\mathbf{M}_1,\mathbf{M}_2))\mathbf{L}_2=(g\cdot R_\mathbf{M}(\mathbf{M}_1,\mathbf{M}_2))\mathbf{L}_2.
\end{equation}
Therefore, $R_\mathbf{M}(g\cdot\mathbf{M}_1,\mathbf{M}_2)=g\cdot R_\mathbf{M}(\mathbf{M}_1,\mathbf{M}_2)$, which implies that $R_\mathbf{M}(\mathbf{M}_1,\mathbf{M}_2)$ is $O(3)$-equivariant.
\end{proof}

\subsubsection{Proof of Proposition 1}
\label{proof_proposition_Unified}

\text{$ $}\newline
\textbf{Proposition 1}. \textit{($O(3)$-Equivariant Tensor Prediction from the Perspective of Canonicalization.)
Given an arbitrary tensor prediction function $f(\mathbf{M}): \mathcal{V} \to\mathcal{W}$, 
we can define a new function $h(\mathbf{M})=R_\mathbf{M}(\mathbf{M},C_\mathbf{M}(\mathbf{M}))\cdot f(C_\mathbf{M}(\mathbf{M}))$,
such that $h(\mathbf{M})$ is $O(3)$ equivariant for tensor prediction. }
\begin{proof}
According to Lemma \ref{lemma_equivariant} and Lemma \ref{lemma_invariant}, for any $g \in O(3)$  we have
\begin{equation}
\begin{aligned}
h(g\cdot\mathbf{M})&=R_\mathbf{M}(g\cdot\mathbf{M},C_\mathbf{M}(g\cdot\mathbf{M}))\cdot f(C_\mathbf{M}(g\cdot\mathbf{M}))\\
&=R_\mathbf{M}(g\cdot\mathbf{M},C_\mathbf{M}(\mathbf{M}))\cdot f(C_\mathbf{M}(\mathbf{M}))\\
&=g\cdot R_\mathbf{M}(\mathbf{M},C_\mathbf{M}(\mathbf{M}))\cdot f(C_\mathbf{M}(\mathbf{M}))\\
&=g\cdot h(\mathbf{M})
\end{aligned}     
\end{equation}
\end{proof}
Thus, $h(\mathbf{M})$ is $O(3)$ equivariant for tensor prediction.

\subsubsection{Proof of Proposition 2}
\label{proof_proposition_plolar_Canonical}
\text{$ $}\newline
\textbf{Proposition 2.}
\textit{(Orbit Canonicalization and Rigid Registration for Crystal Tensor Prediction.) 
By performing polar decomposition on the lattice matrix $\mathbf{L}=\mathbf{Q}\mathbf{H}$ of a crystal $\mathbf{M}$,
we can define functions $f_{p1}(\mathbf{M})=(\mathbf{A}, \mathbf{F}, \mathbf{H})$ and $f_{p2}(\mathbf{M},f_{p1}(\mathbf{M}))=\mathbf{Q}$, where $f_{p1}(\cdot)$ and $f_{p2}(\cdot)$ correspond to performing polar decomposition on the lattice matrix. In this perspective,
$f_{p1}(\cdot)$ serve as the orbit canonicalization function, while $f_{p2}(\cdot)$ serve as the rigid registration function. The matrix $\mathbf{H}$ is thus identified as the canonical form.}

\begin{proof}
Let $\mathbf{L}\in\mathbb{R}^{3\times 3}$ be an invertible lattice matrix.
By the polar decomposition, there exist a unique orthogonal $\mathbf{Q}\in O(3)$ and a unique symmetric positive-definite $\mathbf{H}\in\mathbb{R}^{3\times 3}$ such that
\begin{equation}\label{eq:polar}
\mathbf{L}=\mathbf{Q}\mathbf{H}, 
\qquad 
\big(\mathbf{L}^\top \mathbf{L}\big)^{1/2}=\big(\mathbf{H}^\top \mathbf{H}\big)^{1/2}=\big(\mathbf{H}\mathbf{H}\big)^{1/2}=\mathbf{H}.
\end{equation}

Consider an arbitrary $O(3)$ group transformation $\mathbf{Q'}\in O(3)$ acting on the lattice, i.e., $\mathbf{L}\mapsto \mathbf{Q'}\mathbf{L}$. 
Apply polar decomposition to $\mathbf{Q'}\mathbf{L}$:
\begin{equation}\label{eq:Q1H1}
\mathbf{Q'}\mathbf{L}=\mathbf{Q}_1\mathbf{H}_1,
\qquad 
\mathbf{H}_1=\big((\mathbf{Q'}\mathbf{L})^\top (\mathbf{Q'}\mathbf{L})\big)^{1/2}.
\end{equation}
Since $(\mathbf{Q'})^\top\mathbf{Q'}=\mathbf{I}$, we have
\begin{equation}\label{eq:H-invariance}
\mathbf{H}_1=\big(\mathbf{L}^\top \mathbf{L}\big)^{1/2}=\mathbf{H}.
\end{equation}

Therefore, $\mathbf{H}$ remains unaffected by the $O(3)$ group transformation. Because $\mathbf{A}$ and $\mathbf{F}$ are unchanged by $O(3)$ group transformation and $\mathbf{H}$ is invariant by \eqref{eq:H-invariance}, the triplet
\begin{equation}
\mathbf{M}_0 \coloneqq f_{p1}(\mathbf{M})=(\mathbf{A},\mathbf{F},\mathbf{H})
\end{equation}
is identical for all representatives of the $O(3)$-orbit of $\mathbf{M}$. Therefore, $f_{p1}(\cdot)$ serves as an orbit canonicalization map and the matrix $\mathbf{H}$ is the canonical form.

Furthermore, based on $\mathbf{L} = \mathbf{Q}\mathbf{H} $, 
performing polar decomposition on the lattice matrix to obtain $\mathbf{Q}$, i.e. $f_{p2}(\mathbf{M},\mathbf{M_0})=\mathbf{Q}$, can be regarded as rigid registration function, as shown below:
\begin{equation}
\mathbf{L} = \mathbf{Q}\mathbf{H} \quad\to\quad \mathbf{M} = \mathbf{Q}\mathbf{M_0}\quad\to\quad \mathbf{M} = f_{p2}(\mathbf{M},\mathbf{M_0})\cdot \mathbf{M_0}.
\end{equation}

\end{proof}

\subsection{$O(3)$ equivariance for crystal tensor properties}
\label{tensor_equivariance}

Due to the difference between the input space $\mathcal{V}$ and the output space $\mathcal{W}$, 
the $O(3)$ equivariance required for crystal tensor property prediction tasks differs from the $O(3)$ equivariance typically encountered in general molecular studies \citep{hoogeboom2022equivariant,xuequivariant,song2024equivariant}.
In molecular studies, 
the input space $\mathcal{V}$ and the output space $\mathcal{W}$ are the same, meaning group representations in these spaces are identical, i.e. $\rho_\mathcal{V}(g)=\rho_\mathcal{W}(g)$. Specifically,
for a function $f(\mathbf{X})$, if it is $O(3)$ equivariant, then it satisfies $f(\mathbf{Q}\mathbf{X})=\mathbf{Q}f(\mathbf{X})$, where $\mathbf{Q}\in\mathbb{R}^{3 \times 3}$ is an arbitrary orthogonal matrix, and $\mathbf{X}\in\mathbb{R}^{3 \times N}$ represents the coordinate matrix of atoms in a molecule. In this context, the function $f$ can be seen as the model. 
However, in the prediction of crystal tensor properties, the
requirements for $O(3)$ equivariance differ. 
In what follows, we will provide a detailed introduction to this difference.

\textbf{Dielectric tensor.} 
For the dielectric tensor $\boldsymbol{\varepsilon}\in\mathbb{R}^{3\times3}$ and a function $f:(\mathbf{A},\mathbf{F},\mathbf{L})\to\boldsymbol{\varepsilon}$, if it is $O(3)$ equivariant, it must satisfy $f(\mathbf{A},\mathbf{F},\mathbf{Q}\mathbf{L}) = \mathbf{Q}f(\mathbf{A},\mathbf{F},\mathbf{L})\mathbf{Q}^\top$. The reason for this difference  lies in the physical nature of the dielectric tensor, which characterizes a material's polarization response to an external electric field \citep{yanspace}. Specifically, the dielectric tensor describes the relationship between  the electric displacement $\mathbf{D}\in\mathbb{R}^3$ and the applied electric field $\mathbf{E}\in\mathbb{R}^3$ through the equation $\mathbf{D}=\boldsymbol{\varepsilon}\mathbf{E}$. 
When an $O(3)$ group transformation $\mathbf{Q}$ is applied to the crystal structure, the relationship $\mathbf{D'}=\boldsymbol{\varepsilon}'\mathbf{E'}$ holds, where $\mathbf{D'}=\mathbf{Q}\mathbf{D}$ and $\mathbf{E'}=\mathbf{Q}\mathbf{E}$. Substituting these into the equation, we have:
\begin{equation}
\mathbf{Q}\mathbf{D}=\boldsymbol{\varepsilon}'\mathbf{Q}\mathbf{E}  \to \mathbf{D}=\mathbf{Q}^\top\boldsymbol{\varepsilon}'\mathbf{Q}\mathbf{E},
\end{equation}
which implies that the dielectric tensor transforms under the $O(3)$ group transformation as:
$\boldsymbol{\varepsilon}'=\mathbf{Q}\boldsymbol{\varepsilon}\mathbf{Q}^\top$.
This transformation principle extends similarly to other crystal tensors, as follows.

\textbf{Piezoelectric tensor.} The piezoelectric tensor $\mathbf{e} \in \mathbb{R}^{3\times3\times3}$ describes the relationship between the applied strain $\boldsymbol{\epsilon} \in \mathbb{R}^{3 \times 3}$ to the electric displacement field $\mathbf{D} \in \mathbb{R}^3$ within the material. Mathematically, this relationship is expressed as $\mathbf{D}_i=\sum_{jk}\mathbf{e}_{ijk}\boldsymbol{\epsilon}_{jk}$, with $i, j, k \in \{ 1, 2, 3 \}$. 

When an $O(3)$ group transformation $\mathbf{Q}$ is applied to the crystal, the strain tensor and
electric displacement field is transformed to
$
\boldsymbol{\epsilon}_{jk}^{\prime}=\sum_{mn}\mathbf{Q}_{jm}\mathbf{Q}_{kn}\boldsymbol{\epsilon}_{mn}$ 
and $\mathbf{D}_{i}^{\prime}=\sum_{\ell}\mathbf{Q}_{i\ell}\mathbf{D}_{\ell}$. The relation is then reformulated as $\mathbf{D}_i^\prime=\sum_{jk}\mathbf{e}_{ijk}^\prime\boldsymbol{\epsilon}_{jk}^\prime $.
Since $\mathbf{Q}$ is orthogonal matrix
($\mathbf{Q}^{-1}=\mathbf{Q}^{\top}$), we have
$\boldsymbol{\epsilon}_{jk}=\sum_{mn}\mathbf{Q}_{mj}\mathbf{Q}_{nk}\boldsymbol{\epsilon}^{\prime}_{mn}$. Consequently,  $\mathbf{D}_{i}^{\prime}$ can be represented as
\begin{equation}
\begin{aligned}
\mathbf{D}_{i}^{\prime}&=\sum_{\ell}\mathbf{Q}_{i\ell}\mathbf{D}_{\ell}\\
&=\sum_{\ell}\mathbf{Q}_{i\ell}\sum_{jk}\mathbf{e}_{\ell jk}\boldsymbol{\epsilon}_{jk}\\
&=\sum_{\ell}\mathbf{Q}_{i\ell}\sum_{jk}\mathbf{e}_{\ell jk}(\sum_{mn}\mathbf{Q}_{mj}\mathbf{Q}_{nk}\boldsymbol{\epsilon}^{\prime}_{mn})\\
&=\sum_{\ell}\mathbf{Q}_{i\ell}\sum_{mn}\mathbf{e}_{\ell mn}(\sum_{jk}\mathbf{Q}_{jm}\mathbf{Q}_{kn}\boldsymbol{\epsilon}^{\prime}_{jk}) \quad (exchange \, sign, m \leftrightarrow j, n \leftrightarrow k)\\
&=\sum_{jk}
\sum_{lmn}\mathbf{Q}_{il}\mathbf{Q}_{jm}\mathbf{Q}_{kn}\mathbf{e}_{lmn}
\boldsymbol{\epsilon}_{jk}^\prime 
\end{aligned}
\end{equation}

Therefore, under the $O(3)$ group transformation $\mathbf{Q}$, the transformed piezoelectric tensor $\mathbf{e}_{ijk}^{\prime}$ is given by:
\begin{equation}
\mathbf{e}_{ijk}^{\prime}=\sum_{lmn}\mathbf{Q}_{il}\mathbf{Q}_{jm}\mathbf{Q}_{kn}\mathbf{e}_{lmn}.
\end{equation}

When employing GoeCTP to predict the piezoelectric tensor, Equation \ref{equation11} becomes $\mathbf{e}_{ijk}^{\mathrm{final}}=\sum_{lmn}\mathbf{Q}_{il}\mathbf{Q}_{jm}\mathbf{Q}_{kn}\mathbf{e}_{lmn}$.

\textbf{Elastic tensor.} The elastic tensor $C \in \mathbb{R}^{3\times3\times3\times3}$ describes the relationship between the applied strain $\boldsymbol{\epsilon} \in \mathbb{R}^{3 \times 3}$ and the stress tensor $\sigma \in \mathbb{R}^{3 \times 3}$ within the material. This relationship is expressed as $\boldsymbol{\sigma}_{ij} =\sum_{k \ell}C_{ijk \ell}\boldsymbol{\epsilon}_{k \ell}$, with $i,j,k,\ell \in \{1, 2, 3,4\}$.

When an $O(3)$ group transformation $\mathbf{Q}$ is applied to the crystal, the strain tensor and
stress tensor is transformed to
$
\boldsymbol{\epsilon}_{jk}^{\prime}=\sum_{mn}\mathbf{Q}_{jm}\mathbf{Q}_{kn}\boldsymbol{\epsilon}_{mn}$ 
and $
\boldsymbol{\sigma}_{jk}^{\prime}=\sum_{mn}\mathbf{Q}_{jm}\mathbf{Q}_{kn}\boldsymbol{\sigma}_{mn}$.
Under this transformation, the relationship becomes $\boldsymbol{\sigma}_{ij}^{\prime} =\sum_{k \ell}C_{ijk \ell}^{\prime}\boldsymbol{\epsilon}_{k \ell}^{\prime}$.
Since $\mathbf{Q}$ is orthogonal matrix ($\mathbf{Q}^{-1}=\mathbf{Q}^{\top}$), we have
$\boldsymbol{\epsilon}_{jk}=\sum_{mn}\mathbf{Q}_{mj}\mathbf{Q}_{nk}\boldsymbol{\epsilon}^{\prime}_{mn}$. 
Based on above equations, the transformed stress tensor $\boldsymbol{\sigma}_{ij}^\prime$ can be expressed as:

\begin{equation}
\begin{aligned}
\boldsymbol{\sigma}_{ij}^{\prime} &=\sum_{mn}\mathbf{Q}_{im}\mathbf{Q}_{jn}\boldsymbol{\sigma}_{mn}\\
&=\sum_{mn}\mathbf{Q}_{im}\mathbf{Q}_{jn}\sum_{pq}C_{mnpq}\boldsymbol{\epsilon}_{pq}\\
&=\sum_{mn}\mathbf{Q}_{im}\mathbf{Q}_{jn}\sum_{pq}C_{mnpq}\sum_{k\ell}\mathbf{Q}_{kp}\mathbf{Q}_{\ell q}\boldsymbol{\epsilon}^{\prime}_{k\ell}\\
&=\sum_{k\ell}\sum_{mnpq    }\mathbf{Q}_{im}\mathbf{Q}_{jn}\mathbf{Q}_{kp}\mathbf{Q}_{lq}C_{mnpq}\boldsymbol{\epsilon}^{\prime}_{k\ell}
\end{aligned}
\end{equation}

Therefore, under the $O(3)$ group transformation $\mathbf{Q}$, $C_{ijkl}^{\prime}$ is represented as: 
\begin{equation}
C_{ijkl}^{\prime}=\sum_{mnpq}\mathbf{Q}_{im}\mathbf{Q}_{jn}\mathbf{Q}_{kp}\mathbf{Q}_{lq}C_{mnpq}.
\end{equation}

When GoeCTP is used to predict elastic tensor, Equation \ref{equation11}   becomes
$C_{ijkl}^{\mathrm{final}}=\sum_{mnpq    }\mathbf{Q}_{im}\mathbf{Q}_{jn}\mathbf{Q}_{kp}\mathbf{Q}_{lq}C_{mnpq}$.

\subsection{Tensor properties symmetry}
\label{Voigt}

Tensor properties describe the material's response to external physical fields (such as electric fields or mechanical stress). In materials with symmetry, this response must adhere to the symmetry requirements of the material. As demonstrated by \citet{yanspace}, when the space group transformation $\mathbf{R}$ is applied to the corresponding crystal structure $\mathbf{M} = (\mathbf{A}, \mathbf{F}, \mathbf{L})$, the crystal remains unchanged, i.e., $(\mathbf{A}, \mathbf{F}, \mathbf{L})=(\mathbf{A}, \mathbf{F}, \mathbf{R}\mathbf{L})$. Therefore, the corresponding tensor properties also remain unchanged, i.e., $\boldsymbol{\varepsilon}=\mathbf{R}\boldsymbol{\varepsilon}\mathbf{R}^\top$. Thus, crystal symmetry imposes strict constraints on the components of the tensor, leading to the simplification or elimination of many components, reducing the number of independent components. 

\begin{table}[h]
    \centering
    \scalebox{0.95}{
    \begin{tabular}{c c c }
        \toprule

        Crystal system&  Number of independent elements&  Dielectric tensor \\       
        \midrule
        Cubic& 1&$\boldsymbol{\varepsilon}=\begin{pmatrix}\boldsymbol{\varepsilon}_{11}&0&0\\0&\boldsymbol{\varepsilon}_{11}&0\\0&0&\boldsymbol{\varepsilon}_{11}\end{pmatrix}$ \\
        Tetragonal \& Hexagonal \&  Trigonal & 2&$\boldsymbol{\varepsilon}=\begin{pmatrix}\boldsymbol{\varepsilon}_{11}&0&0\\0&\boldsymbol{\varepsilon}_{11}&0\\0&0&\boldsymbol{\varepsilon}_{33}\end{pmatrix}$  \\

        Orthorhombic& 3& $\boldsymbol{\varepsilon}=\begin{pmatrix}\boldsymbol{\varepsilon}_{11}&0&0\\0&\boldsymbol{\varepsilon}_{22}&0\\0&0&\boldsymbol{\varepsilon}_{33}\end{pmatrix}$\\
        
        Monoclinic& 4& $\boldsymbol{\varepsilon}=\begin{pmatrix}\boldsymbol{\varepsilon}_{11}&0&\boldsymbol{\varepsilon}_{13}\\0&\boldsymbol{\varepsilon}_{22}&0\\\boldsymbol{\varepsilon}_{13}&0&\boldsymbol{\varepsilon}_{33}\end{pmatrix}$\\
        
        Triclinic& 6& $\boldsymbol{\varepsilon}=\begin{pmatrix}\boldsymbol{\varepsilon}_{11}&\boldsymbol{\varepsilon}_{12}&\boldsymbol{\varepsilon}_{13}\\\boldsymbol{\varepsilon}_{12}&\boldsymbol{\varepsilon}_{22}&\boldsymbol{\varepsilon}_{23}\\\boldsymbol{\varepsilon}_{13}&\boldsymbol{\varepsilon}_{23}&\boldsymbol{\varepsilon}_{33}\end{pmatrix}$\\        
        \bottomrule
    \end{tabular}
    }
    \caption{Number of independent components in the dielectric tensor for different crystal systems.}
        \label{dielectric_crystal_systems}
\end{table}

\textbf{Independent components in the dielectric tensor.}
The $3\times3$ dielectric tensor has a minimum of 1 and a maximum of 6 independent elements for various types of systems due to the crystal symmetry.
The number of independent components in the dielectric tensor for different crystal systems is shown in Tabel \ref{dielectric_crystal_systems} \citep{mao2024dielectric}.

\textbf{Voigt notation for elastic tensor.} Similar to the dielectric tensor,
the elastic tensor also
has independent elements for various types of systems due to the crystal symmetry.
The elastic tensor has a minimum of 3 and a maximum of 21 independent elements for various types of systems.
Voigt notation is a compact way to represent these independent components of tensor properties \citep{itin2013constitutive}.
According to the rules 
$11\to1\mathrm{~;~}22\to2\mathrm{~;~}33\to3\mathrm{~;~}23,~32\to4\mathrm{~;~}31,~13\to5\mathrm{~;~}12,~21\to6$,
the elastic tensor in Voigt notation is a $6\times6$ symmetric matrix \citep{wen2024equivariant,ran2023velas}:
\begin{equation}
\scalebox{0.9}{$
C=\begin{pmatrix}
C_{1111}&C_{1122}&C_{1133}&C_{1123}&C_{1131}&C_{1112}
\\C_{1122}&C_{2222}&C_{2233}&C_{2223}&C_{2231}&C_{2212}
\\C_{1133}&C_{2233}&C_{3333}&C_{3323}&C_{3331}&C_{3312}
\\C_{1123}&C_{2223}&C_{3323}&C_{2323}&C_{2331}&C_{2312}
\\C_{1131}&C_{2231}&C_{3331}&C_{2331}&C_{3131}&C_{3112}
\\C_{1112}&C_{2212}&C_{3312}&C_{2312}&C_{3112}&C_{1212}\end{pmatrix}\to\begin{pmatrix}C_{11}&C_{12}&C_{13}&C_{14}&C_{15}&C_{16}\\C_{12}&C_{22}&C_{23}&C_{24}&C_{25}&C_{26}\\C_{13}&C_{23}&C_{33}&C_{34}&C_{35}&C_{36}\\C_{14}&C_{24}&C_{34}&C_{44}&C_{45}&C_{46}\\C_{15}&C_{25}&C_{35}&C_{45}&C_{55}&C_{56}\\C_{16}&C_{26}&C_{36}&C_{46}&C_{56}&C_{66}\end{pmatrix}$
}
\end{equation}
The number of independent components in the elastic tensor for different crystal systems is shown in Tabel \ref{elastic_crystal_systems} (partial data shown; for more details, refer to \citet{wen2024equivariant,ran2023velas}).

\begin{table}[h]
    \centering
    \scalebox{0.8}{
    \begin{tabular}{c c c }
        \toprule

        Crystal system&  Number of independent elements&  Elastic tensor \\       
        \midrule
        Cubic& 3&$\begin{gathered}C=\begin{pmatrix}C_{11}&C_{12}&C_{12}&0&0&0\\C_{12}&C_{11}&C_{12}&0&0&0\\C_{12}&C_{12}&C_{11}&0&0&0\\0&0&0&C_{44}&0&0\\0&0&0&0&C_{44}&0\\0&0&0&0&0&C_{44}\end{pmatrix}\end{gathered}$ \\
        


        Tetragonal  & 6& $\begin{gathered}C=\begin{pmatrix}C_{11}&C_{12}&C_{13}&0&0&0\\C_{12}&C_{11}&C_{13}&0&0&0\\C_{13}&C_{13}&C_{33}&0&0&0\\0&0&0&C_{44}&0&0\\0&0&0&0&C_{44}&0\\0&0&0&0&0&C_{66}\end{pmatrix}\end{gathered}$\\

        
        
        Triclinic& 21& $\begin{gathered}C=\begin{pmatrix}C_{11}&C_{12}&C_{13}&C_{14}&C_{15}&C_{16}\\C_{12}&C_{22}&C_{23}&C_{24}&C_{25}&C_{26}\\C_{13}&C_{23}&C_{33}&C_{34}&C_{35}&C_{36}\\C_{14}&C_{24}&C_{34}&C_{44}&C_{45}&C_{46}\\C_{15}&C_{25}&C_{35}&C_{45}&C_{55}&C_{56}\\C_{16}&C_{26}&C_{36}&C_{46}&C_{56}&C_{66}\end{pmatrix}\end{gathered}$\\        
        \bottomrule
    \end{tabular}
    }
    \caption{Number of independent components in the elastic tensor for different crystal systems.}
        \label{elastic_crystal_systems}
\end{table}

\textbf{Voigt notation for piezoelectric tensor.} 
The number of independent components in the piezoelectric tensor for different crystal systems is shown in Tabel \ref{piezoelectric_crystal_systems} (partial data shown; for more details, refer to \citet{de2015database,gorfman2024piezoelectric}).

\begin{table}[h]
    \centering
    \scalebox{0.8}{
    \begin{tabular}{c c c c }
        \toprule

        Crystal system& point groups& Number of independent elements&  Piezoelectric tensor \\       
        \midrule
        

        Trigonal &32 &2 & $\mathbf{e}=\begin{pmatrix}\mathbf{e}_{11}&-\mathbf{e}_{11}&0&\mathbf{e}_{14}&0&0\\0&0&0&0&-\mathbf{e}_{14}&-\mathbf{e}_{11}\\0&0&0&0&0&0\end{pmatrix}$\\


        
        Monoclinic& 2&8& $\mathbf{e}=\begin{pmatrix}0&0&0&\mathbf{e}_{14}&0&\mathbf{e}_{16}\\\mathbf{e}_{21}&\mathbf{e}_{22}&\mathbf{e}_{23}&0&e_{25}&0\\0&0&0&\mathbf{e}_{34}&0&\mathbf{e}_{36}\end{pmatrix}$\\
        
        Triclinic& 1& 18& $\mathbf{e}=\begin{pmatrix}
\mathbf{e}_{11} & \mathbf{e}_{12} & \mathbf{e}_{13} & \mathbf{e}_{14} & \mathbf{e}_{15} & \mathbf{e}_{16} \\
\mathbf{e}_{21} & \mathbf{e}_{22} & \mathbf{e}_{23} & \mathbf{e}_{24} & \mathbf{e}_{25} & \mathbf{e}_{26} \\
\mathbf{e}_{31} & \mathbf{e}_{32} & \mathbf{e}_{33} & \mathbf{e}_{34} & \mathbf{e}_{35} & \mathbf{e}_{36}
\end{pmatrix}$\\        
        \bottomrule
    \end{tabular}
    }
    
    \caption{Number of independent components in the piezoelectric tensor for different crystal systems.}
\label{piezoelectric_crystal_systems}
    
\end{table}

\subsection{Experimental details}
\label{Experimental_details}


\textbf{Datasets details.} 
The dataset for dielectric tensor and piezoelectric tensor is derived from the data processed by \citet{yanspace}, sourced from the JARVIS-DFT database. 
Since nearly half of the crystal samples in the piezoelectric dataset have piezoelectric tensor labels that are entirely zero, training {\em GoeCTP} directly on this dataset can lead to severe overfitting, thereby compromising the reliability of the experimental results.
Therefore, we filtered out all crystal samples with zero-valued piezoelectric tensor labels from piezoelectric dataset.
For the elastic tensor prediction task, the dataset is also obtained from the JARVIS-DFT database and is publicly accessible as \texttt{dft\_3d} through the \texttt{jarvis-tools} package \footnote{\url{https://pages.nist.gov/jarvis/databases/}}. The statistical details of the datasets are presented in Table \ref{Dataset_statistics2}.
Additionally, 
in the dielectric tensor dataset, the dielectric tensor is a $3\times3$ symmetric matrix. Therefore, during prediction, we predict 6 elements of the matrix (see Appendix \ref{Voigt} for details) and then reconstruct the entire $3\times3$ symmetric matrix.
In the piezoelectric tensor dataset, the piezoelectric tensor is represented using Voigt notation as a $3\times6$ matrix, rather than a $3\times3\times3$ third-order tensor
(see Appendix \ref{Voigt} for details).
we only predict the $3\times6$ matrix.
In the elastic tensor dataset, the elastic tensor is represented using Voigt notation as a $6\times6$ symmetric matrix, rather than a $3\times3\times3\times3$ fourth-order tensor \citep{itin2013constitutive,wen2024equivariant}. Therefore, we predict the $6\times6$ matrix during usage (see Appendix \ref{Voigt} for details).

\begin{table}[h]
    \centering
\resizebox{0.48\textwidth}{!}{ 
    \begin{tabular}{c|c |c |c|c }
        \toprule

        Dataset &Sample size &Fnorm Mean  & Fnorm STD  &Unit \\       
        \midrule
        Dielectric& 4713 &14.7 & 18.2 & Unitless \\

        Piezoelectric&  2701 & 0.79 & 4.03 & $\text{C}/\text{m}^2$ \\
        Elastic& 25110  & 306.4 & 238.4 & GPa \\

        \bottomrule
    \end{tabular}
    }
    \caption{Dataset statistics.}
       \label{Dataset_statistics2}
\end{table}

\textbf{Experimental Settings.}
The experiments were performed using a single NVIDIA GeForce RTX 3090 GPU. For benchmarking, we directly utilized the codebases for ETGNN and GMTNet as provided by \citet{yanspace}. For each property, the dataset is split into training, validation, and test sets in an 8:1:1 ratio \citep{yanspace}.

\textbf{Hyperparameter settings of {\em GoeCTP (eCom.)}.}
When constructing the crystal graph, we used the 16th nearest atom to determine the cutoff radius. For edge embeddings, we used an RBF kernel with $c=0.75$
and values ranging from $-$4 to 0, 
which maps $-c/||\mathbf{e}_{ij}||_2$ to a 512-dimensional vector. In the dielectric and piezoelectric tensor prediction task, the 512-dimensional vector is mapped to a 128-dimensional vector through a non-linear layer, while in the elastic tensor prediction, it is mapped to a 256-dimensional vector.
For the prediction module, eComFormer \citep{yancomplete}, in the dielectric tensor prediction task, four node-wise transformer layers and one node-wise equivariant updating layer were utilized. For the piezoelectric and elastic tensor prediction tasks, two node-wise transformer layers and one node-wise equivariant updating layer were used. The learning rate was set to 0.001 for both dielectric and elastic tensor tasks, trained for 200 epochs with a batch size of 64. For the piezoelectric tensor task, the learning rate was set to 0.002 under the same training configuration.
In the Reverse R\&R module, for the dielectric tensor prediction task, the 128-dimensional feature vector output from the prediction module was mapped to a 6-dimensional vector through two nonlinear layers, then reshaped into a $3\times3$ matrix and combined with the orthogonal matrix $\mathbf{Q}$ obtained from the R\&R module to ensure $O(3)$-equivariant output.
For the piezoelectric tensor prediction task, the 18-dimensional feature vector output from the prediction module was reconstructed into a $3\times6$ matrix and combined with the orthogonal matrix $\mathbf{Q}$ from the R\&R module to achieve $O(3)$-equivariant output.
For the elastic tensor prediction task, the 256-dimensional feature vector output from the prediction module was mapped to a 36-dimensional vector through two nonlinear layers, then reconstructed into a $6\times6$ matrix and combined with the orthogonal matrix $\mathbf{Q}$ from the R\&R module to achieve $O(3)$-equivariant output.
During model training, the AdamW optimizer \citep{loshchilovdecoupled} was employed with a weight decay of $10^{-5}$ and a polynomial learning rate decay schedule.
For the dielectric and elastic tensor prediction tasks, the Huber loss function was used, while for the piezoelectric tensor prediction task, the L1 loss function was adopted.

\textbf{Hyperparameter settings of {\em GoeCTP (C. Fra.)}.}
For all tensor prediction tasks, the configuration of the prediction module, CrystalFramer \citep{ito2025rethinking}, including the input atom-embedding layer, self-attention blocks, and other architectural components, was kept identical to the setup described in the original work \citep{ito2025rethinking} on the JARVIS dataset.
For the dielectric tensor prediction task, the model was trained using an L1 loss function with a learning rate of 0.002, a batch size of 256, and 200 training epochs, following a polynomial learning rate decay schedule.
For the piezoelectric tensor prediction task, the model employed an L1 loss function with a learning rate of 0.002, a batch size of 512, and 1000 training epochs, using a warm-up-free inverse square root schedule \citep{huang2020improving,ito2025rethinking}.
For the elastic tensor prediction task, the model adopted the Huber loss function with a learning rate of 0.001, a batch size of 64, and 500 training epochs, along with a polynomial learning rate decay schedule.
Across all experiments, the AdamW optimizer \citep{loshchilovdecoupled} was used with parameters $(\beta_1, \beta_2) = (0.9, 0.98)$ and a weight decay of $10^{-5}$.

\textbf{Hyperparameter settings of GMTNet and ETGNN.}
Following \citet{yanspace},
we trained GMTNet and ETGNN for 200 epochs using Huber loss with a learning rate of 0.001 and Adam optimizer with $10^{-5}$ weight decay across all tasks. The same polynomial learning rate decay scheduler is used in all experiments. 
For the dielectric tensor prediction task and the piezoelectric tensor prediction task, the network configuration of GMTNet followed that described in the original paper, comprising two layers for node-invariant feature updating and three layers for equivariant message passing. For the elastic tensor prediction task, the configuration was adjusted to include two layers for node-invariant feature updating and five layers for equivariant message passing.

\subsection{Additional results.}
\label{5_times}

We assessed the performance of GoeCTP when combined with iComFormer \citep{yancomplete} and CrystalFormer \citep{taniaicrystalformer} on both the dielectric and elastic tensor datasets. 
We additionally included the baseline method MEGNet \citep{chen2019graph,morita2020modeling} in our comparative experiments.
Additional comparison results for the dielectric tensor dataset are presented in Table \ref{ec_ic_outcome},
and those for the elastic tensor dataset are presented in Table \ref{ec_ic_outcome2}.

\begin{table}[t!]
    \centering
    \scalebox{0.825}{
    \begin{tabular}{c|c cc|c c c c}
        \toprule

        &MEGNET& ETGNN & GMTNet  &{\bf GoeCTP (eCom.)} &{\bf GoeCTP (iCom.)} &{\bf GoeCTP (Crys.)}& {\bf GoeCTP (C. Fra.)}  \\      
        \midrule
        Fnorm $\downarrow$& 3.71 & 3.40 & 3.28 & 3.23&3.40&3.53 & \textbf{3.05}\\

        EwT 25\% $\uparrow$& 75.8\%  & 82.6\% & 83.3\% & 83.2\% &81.7\%&80.1\%& \textbf{86.4\%}\\
        EwT 10\% $\uparrow$ & 38.9\% & 49.1\% & 56.0\% & 56.8\% &53.8\%&52.9\% &\textbf{62.6\%}\\
        EwT 5\% $\uparrow$ & 18.0\% & 25.3\% & 30.5\% & 35.5\%&32.3\%&30.6\%&\textbf{43.5\%}\\

        Total Time (s) $\downarrow$& 663 & 1325 & 1611 & 616&\textbf{535} & 645&1976\\

        \bottomrule
    \end{tabular}
    }
    \caption{Additional comparison of performance metrics on dielectric dataset.}

        \label{ec_ic_outcome}
\end{table}

\begin{table}[t!]
    \centering
    \scalebox{0.825}{
    \begin{tabular}{c|c cc|c c c c}
        \toprule

        &MEGNET& ETGNN & GMTNet  &{\bf GoeCTP (eCom.)} &{\bf GoeCTP (iCom.)} &{\bf GoeCTP (Crys.)}& {\bf GoeCTP (C. Fra.)}  \\       
        \midrule
        Fnorm $\downarrow$& 143.86 & 123.64 & 117.62 &107.11& 102.80 &107.44 &\textbf{95.98}\\

        EwT 25\% $\uparrow$& 23.6\%  & 32.0\% & 36.0\% & 42.5\%& 46.7\%& 43.5\% &\textbf{49.7\%}\\
        EwT 10\% $\uparrow$ & 3.0\% & 3.8\% & 7.6\% & 15.3\% &18.6\%&15.8\%&\textbf{20.3\%}\\
        EwT 5\% $\uparrow$ & 0.5\% & 0.5\% & 2.0\% &7.2\%&8.2\% &7.9\% &\textbf{9.8\%}\\

        Total Time (s) $\downarrow$& 2899 & 4448 & $>$ 36000 & \textbf{2422}&4035 & 7891 &17197\\

        \bottomrule
    \end{tabular}
    }
    \caption{Additional comparison of performance metrics on 
elastic dataset.}
    \label{ec_ic_outcome2}
\end{table}

\subsection{How to further utilize the tensor properties symmetry}
\label{symmetry_utilize}

In this section, we will use the dielectric tensor as an example to briefly discuss how to further utilize the symmetry of tensor properties.

\textbf{Utilizing symmetry constraints for zero elements.} To demonstrate the role of symmetry constraints, we present an example of the GoeCTP (eCom.) prediction results in Table \ref{example_GoeCTP1}. Within a reasonable margin of error, the predictions align well with the expected constraints.

\begin{table}[!h]
    \centering
    \scalebox{0.95}{
    \begin{tabular}{c| c| c }
        \toprule

          Label& Prediction& Cubic dielectric tensor \\       
        \midrule
        
        $\begin{pmatrix}2.258&0&0\\0&2.258&0\\0&0&2.258\end{pmatrix}$  &$\begin{pmatrix}2.252&0.016&0.008\\0.016&2.230&0.007\\0.008&0.007&2.262\end{pmatrix}$ 
        &$\boldsymbol{\varepsilon}=\begin{pmatrix}\boldsymbol{\varepsilon}_{11}&0&0\\0&\boldsymbol{\varepsilon}_{11}&0\\0&0&\boldsymbol{\varepsilon}_{11}\end{pmatrix}$ \\
 
        \bottomrule
    \end{tabular}
    }
    \caption{An example of GoeCTP prediction results}
\label{example_GoeCTP1}
\end{table}

\begin{table}[!h]
    \centering
    \scalebox{0.95}{
    \begin{tabular}{c| c|  c|  c|  c|c}
        \toprule

        Crystal system&  Cubic&  Tetragonal&  Hexagonal-Trigonal &Orthorhombic&Monoclinic\\ 
        
        \midrule
        Success rate& 88.3\% & 86.6\%  &84.5\%&84.5\%&75.7\%\\        
        \bottomrule
    \end{tabular}
    }
    \caption{The GoeCTP results of predicting symmetry-constrained zero-valued dielectric
tensor elements.}
\label{GoeCTP_zero_rate} 
\end{table}

For a dielectric tensor, we evaluate the prediction of zero elements by defining a threshold as 1\% of the average value of non-zero elements in the labels. Predictions are deemed successful if they meet this threshold. The results are as shown in Table \ref{GoeCTP_zero_rate}. 
It can be observed that our method predicts most zero elements accurately, though not flawlessly. Since our advantage lies in transferring equivariance through an external framework, there are no restrictions on the model itself. Therefore, to achieve a higher success rate for zero elements, we added a ReLU activation function to the output layer of the network to improve the success rate 
(this applies only to cases where tensor elements are greater than or equal to zero).
The results after retraining GoeCTP (eCom.) are as shown in Table \ref{GoeCTP_ReLU}, Table \ref{GoeCTP_ReLU2}, and 
Table \ref{GoeCTP_ReLU3}. Our method predicts zero elements more accurately while maintaining overall performance.

\begin{table}[!h]
    \centering
    \scalebox{0.95}{
    \begin{tabular}{c| c| c }
        \toprule

          Label & Prediction& Cubic dielectric tensor \\       
        \midrule
        
        $\begin{pmatrix}2.258&0&0\\0&2.258&0\\0&0&2.258\end{pmatrix}$  &$\begin{pmatrix}2.237&0.000&0.000\\0.000&2.283&0.000\\0.000&0.000&2.228\end{pmatrix}$ 
        &$\boldsymbol{\varepsilon}=\begin{pmatrix}\boldsymbol{\varepsilon}_{11}&0&0\\0&\boldsymbol{\varepsilon}_{11}&0\\0&0&\boldsymbol{\varepsilon}_{11}\end{pmatrix}$ \\
 
        \bottomrule
    \end{tabular}
    }
    \caption{An example of GoeCTP (ReLU) prediction results}
\label{GoeCTP_ReLU}
\end{table}

\begin{table}[!h]
    \centering
    \scalebox{0.95}{
    \begin{tabular}{c| c|  c|  c|   c|c}
        \toprule

        Crystal system&  Cubic&  Tetragonal&  Hexagonal-Trigonal &Orthorhombic&Monoclinic\\ 
        
        \midrule
        Success rate& 100\% & 100\%  &87.2\%&100\%&100\%\\        
        \bottomrule
    \end{tabular}
    }
    \caption{The GoeCTP (ReLU) results of predicting symmetry-constrained zero-valued dielectric
tensor elements.}
\label{GoeCTP_ReLU2}        
\end{table}

\begin{table}[!h]
    \centering
    \begin{tabular}{c| c|c }
        \toprule
      
        &GoeCTP &GoeCTP (ReLU) \\       
        \midrule
        Fnorm $\downarrow$& 3.23 & 3.26\\

        EwT 25\% $\uparrow$& 83.2\%  &82.6\% \\
        EwT 10\% $\uparrow$ & 56.8\%&58.4\%\\
        EwT 5\% $\uparrow$  &35.5\%&36.3\% \\

        \bottomrule
    \end{tabular}
    \caption{Performance comparisons between  GoeCTP and GoeCTP (ReLU) on the dielectric dataset.}
\label{GoeCTP_ReLU3}     
\end{table}

This simple modification enables GoeCTP (eCom.) to more accurately predict the zero elements in dielectric tensors caused by the space group constraints.

\textbf{Utilizing symmetry constraints for non-zero elements.}
The symmetry constraints for zero elements serve as a foundational example. When the space group of the input crystal is known, prior knowledge (refer to Appendix \ref{Voigt}) can be employed to ensure that network outputs fully adhere to tensor property constraints.
There are two specific methods to achieve this.

In the first method, only the independent components of the tensor properties are predicted. For example, for a cubic crystal, we predict only one independent component and then use the mask from Table \ref{dielectric_crystal_systems} to reconstruct it into a $3 \times 3$ tensor). In Table \ref{example_GoeCTP1}, GoeCTP (eCom.) was modified to directly predict independent components. The new results are shown in Table \ref{newexample_GoeCTP}, which fully satisfy the symmetry constraints. As shown in Table \ref{GoeCTP_mask3}, the overall prediction performance indicates that directly predicting the independent components improves the model's Fnorm performance on the dielectric
dataset. 
Since this method directly predicts independent components, it is more suitable for cases like the dielectric tensor, where the independent components exhibit relatively simple structures. However, for more complex datasets, such as those involving piezoelectric and elastic tensors, the model's performance using this method may degrade. In such scenarios, the second method provides a more viable alternative.

\begin{table}[!h]
    \centering
    \scalebox{0.95}{
    \begin{tabular}{c| c| c }
        \toprule

          Label & Prediction& Cubic dielectric tensor \\       
        \midrule
        
        $\begin{pmatrix}2.258&0&0\\0&2.258&0\\0&0&2.258\end{pmatrix}$  &$\begin{pmatrix}2.357&0.000&0.000\\0.000&2.357&0.000\\0.000&0.000&2.357\end{pmatrix}$ 
        &$\boldsymbol{\varepsilon}=\begin{pmatrix}\boldsymbol{\varepsilon}_{11}&0&0\\0&\boldsymbol{\varepsilon}_{11}&0\\0&0&\boldsymbol{\varepsilon}_{11}\end{pmatrix}$ \\
 
        \bottomrule
    \end{tabular}
    }
    \caption{An example of GoeCTP (predicting the independent components) prediction results}
\label{newexample_GoeCTP}
\end{table}

\begin{table}[!h]
    \centering
  
    \begin{tabular}{c| c|c }
        \toprule
      
        &GoeCTP &GoeCTP (independent) \\       
        \midrule
        Fnorm $\downarrow$& 3.23 & 3.11\\

        EwT 25\% $\uparrow$& 83.2\%  &80.5\% \\
        EwT 10\% $\uparrow$ & 56.8\%&57.3\%\\
        EwT 5\% $\uparrow$  &35.5\%&35.5\% \\

        \bottomrule
    \end{tabular}
    \caption{Performance comparisons between  GoeCTP and GoeCTP (predicting the independent components) on the dielectric dataset.}
\label{GoeCTP_mask3}     
\end{table}

In the second method, once the network training is completed, the network output is weighted using the standard tensor property format, referred to as a mask, as detailed in Table \ref{dielectric_crystal_systems}. The application of this method to the elastic datasets is demonstrated in Table \ref{GoeCTP_eas_mask}. The results indicate that incorporating symmetry-based weighting effectively improves the performance of GoeCTP.

\begin{table}[!h]
    \centering
    
    \begin{tabular}{c| c|c }
        \toprule
      
        &GoeCTP &GoeCTP (weighted) \\       
        \midrule
        Fnorm $\downarrow$& 107.11 & 106.24\\

        EwT 25\% $\uparrow$& 42.5\%  &46.3\% \\
        EwT 10\% $\uparrow$ & 15.3\%&16.5\%\\
        EwT 5\% $\uparrow$  &7.2\%&7.4\% \\

        \bottomrule
    \end{tabular}
    \caption{Performance comparisons between  GoeCTP and GoeCTP (weighted by the standard tensor property form) on the elastic dataset.}
\label{GoeCTP_eas_mask}     
\end{table}

\subsection{Analysis of discontinuity in canonicalization 
}
\label{non_continuous}

Previous studies \citep{dymequivariant} have shown that canonicalization can often introduce discontinuities into the model. To assess the impact of such discontinuities on our method, we analyze our proposed method from both theoretical and empirical perspectives.

\textbf{From theoretical perspective:} 
The following proposition, presented in Appendix B of \citet{dymequivariant}, characterizes the existence of a continuous canonicalization under orthogonal group actions:

\begin{proposition}
Consider the action of the orthogonal group $O(d)$ on $\mathbf{X}\in\mathbb{R}^{n \times d}$. There exists a continuous canonicalization for $O(d)$ when $n \leq d$. This continuous canonicalization can be defined as $y_\mathbf{X} = \left( \mathbf{X}^\top \mathbf{X} \right)^{1/2}$, where $\mathbf{X}^\top \mathbf{X}$ forms a positive semi-definite matrix and the square root is the standard square root of a positive semi-definite matrix, i.e. $\left( \mathbf{A} \right)^{1/2}$ is the unique positive semidefinite matrix $\mathbf{B}$ such that $\mathbf{B}^2 = \mathbf{A}$.
\label{proposition_noncon}
\end{proposition}

In fact, the polar decomposition adopted in this work can be regarded as a case of such continuous canonicalization.

\begin{proof}
According to Proposition \ref{proposition_noncon}, for a matrix $\mathbf{L} \in\mathbb{R}^{3 \times 3}$, its continuous canonicalization norm is given by $y_\mathbf{L} = \left( \mathbf{L}^\top \mathbf{L} \right)^{1/2}$.   

In this work, the polar decomposition is applied to the lattice matrix $\mathbf{L}=\mathbf{Q}\mathbf{H}, \mathbf{L} \in\mathbb{R}^{3 \times 3}$, where $\mathbf{Q}$ is an orthogonal matrix, and $\mathbf{H}$ is the positive semi-definite Gram matrix (which serves as the canonicalization norm). Based on the properties of positive semi-definite matrices, our method is actually equivalent to the canonicalization approach described in Proposition \ref{proposition_noncon}, i.e. $y_\mathbf{L} = \left( \mathbf{L}^\top \mathbf{L} \right)^{1/2}=\left( (\mathbf{Q}\mathbf{H})^\top \mathbf{Q}\mathbf{H} \right)^{1/2}=\left( \mathbf{H}^\top \mathbf{H} \right)^{1/2}=\mathbf{H}$.
\end{proof}

Therefore, our canonicalization via polar decomposition is consistent with the continuous canonicalization. 

\textbf{From empirical perspective:}
To further demonstrate how this discontinuity affects the accuracy of tensor property predictions, we conducted a set of comparative experiments in which we replaced the polar decomposition in GoeCTP (eCom.) with QR decomposition \citep{linequivariance} for evaluation. 
In this QR decomposition case, the lattice matrix is decomposed as $\mathbf{L}=\mathbf{Q}\mathbf{R}$, where $\mathbf{Q}$ is an orthogonal matrix, and $\mathbf{R}$ is an upper triangular matrix that serves as the canonicalization norm. However, $\mathbf{R}$ does not satisfy Proposition \ref{proposition_noncon}, and therefore does not constitute a continuous canonicalization approach.

To assess the sensitivity of each approach to small structural perturbations, we used the crystal $CdI_2$ (JARVIS\_ID: JVASP-29631) as the model input, and then perturbed its lattice matrix (specifically, scaling down the first row of the lattice matrix) to observe the relative changes in the output (specifically the variations in the Fabenius norm of the output tensor). The detailed results are as follows:

\begin{table}[!h]
    \centering
    \resizebox{\textwidth}{!}{
    \begin{tabular}{l | c  c  c  c  c  c  c  c   c c}
        \toprule
        Perturbation Ratio&0\%&5\% &10\% &15\% &20\% &25\% & 30\%  & 35\%  & 40\% \\       
        \midrule
        Output variation ratio (\%): {\bf GoeCTP (QR)}&0&3.38&6.63&11.29&21.83&39.37& 67.81 & 124.79 & 191.28 \\
Output variation ratio (\%): {\bf GoeCTP (polar)}&0&2.53 &1.23&2.03&4.84&12.72& 24.88 &36.34 & 45.67 \\
       
        \bottomrule
    \end{tabular}
    }
    \caption{Output variation ratio under different perturbation ratios on the dielectric dataset}
\end{table}

The results clearly show that the output of GoeCTP using polar decomposition exhibits significantly smaller relative changes compared to the QR-based variant, particularly under increasing levels of perturbation. This highlights the robustness of the continuous canonicalization afforded by polar decomposition, demonstrating its superiority in maintaining prediction stability and structural consistency under small input variations.

\subsection{How to design other canonical forms}
\label{local_frame_design}

In fact, the canonical form obtained via polar decomposition provides a foundation from which infinitely many alternative canonical forms can be derived. Specifically, given a crystal with lattice matrix $\mathbf{L}$, polar decomposition yields $\mathbf{L}=\mathbf{Q}\mathbf{H}$. 
While $\mathbf{H}$ serves as one valid canonical representation, additional canonical forms can be constructed by applying orthogonal transformations to $\mathbf{H}$.
For instance, given a fixed orthogonal matrix $\mathbf{Q}_0$ (which may be arbitrarily chosen) or a neural network $f_Q(\mathbf{H})=\mathbf{Q}_\theta$ that predicts an orthogonal matrix $\mathbf{Q}_\theta$, one can construct new canonical forms such as $\mathbf{Q}_0\mathbf{H}$ or $\mathbf{Q}_\theta\mathbf{H}$. 
This is because, for any given crystal with an arbitrary orientation, a fixed canonical form can be obtained according to a predefined computation rule, i.e. $\mathbf{L}=\mathbf{Q}\mathbf{H} \quad\to\quad \mathbf{L}=\mathbf{Q}\mathbf{Q}_0^\top\mathbf{Q}_0\mathbf{H} \quad\text{or}\quad\mathbf{L}=\mathbf{Q}\mathbf{Q}_\theta^\top\mathbf{Q}_\theta\mathbf{H}$.
This formulation highlights the inherent flexibility in defining canonical forms.

\subsection{Impact of different canonical form on tensor property prediction}
\label{local_frame}

As discussed in Appendix \ref{local_frame_design}, there exist infinitely many possible canonical forms. To investigate the impact of different canonical forms on tensor property prediction in crystalline materials, we first present the results comparing the performance of GoeCTP (eCom.) with QR decomposition \citep{linequivariance} and GoeCTP (eCom.) without canonicalization. The corresponding results are summarized in the table below:

\begin{table}[h!]
    \centering
    \begin{tabular}{c| c|c|c }
        \toprule
      
        &{\bf GoeCTP (QR)} &{\bf GoeCTP (polar)} &{\bf GoeCTP (w/o QR or polar)}\\       
        \midrule
        Fnorm $\downarrow$& \textbf{3.20} & 3.23& 3.60 \\

        EwT 25\% $\uparrow$&\textbf{83.5\% }&83.2\%& 80.6\% \\
        EwT 10\% $\uparrow$ &56.0\%&\textbf{56.8\%}&56.2\% \\
        EwT 5\% $\uparrow$ & 32.4\%&\textbf{35.5\%}& 32.6\% \\

        Total Time (s) $\downarrow$ & 639&616&613\\
       
        \bottomrule
    \end{tabular}
    \caption{Predictive performance comparisons between GoeCTP (QR), GoeCTP (polar), and GoeCTP (w/o QR or polar) on the dielectric dataset.}
\end{table}

According to the results, it is evident that incorporating a canonical form improves tensor prediction performance compared to omitting it. While different canonical forms lead to slight variations in performance, the overall impact remains relatively modest. One possible explanation is that canonical forms typically exhibit structured representations, such as the symmetric positive semi-definite matrix obtained via polar decomposition or the upper triangular matrix from QR decomposition, which may facilitate the model's ability to capture relevant features.

\end{document}